\documentclass[preprint,12pt]{elsarticle}

\usepackage[utf8]{inputenc}
\usepackage[english]{babel}
\usepackage{amsmath,amssymb,amsfonts,bm}
\usepackage{graphicx}
\usepackage{booktabs}
\usepackage{stmaryrd}
\usepackage{tabularx}
\usepackage{array}
\usepackage{tikz}
\usepackage{subcaption}
\usepackage{epstopdf}
\usepackage{verbatim}
\usepackage{placeins}
\usepackage{amsthm}
\usepackage{enumitem}
\usepackage[toc,page]{appendix}
\usetikzlibrary{arrows.meta,positioning}

\usepackage{algorithm}
\usepackage{algpseudocode}
\usepackage{caption}
\usepackage{siunitx}
\usepackage{xcolor}  

\newtheorem{theorem}{Theorem}[section]
\newtheorem{definition}{Definition}[section]
\newtheorem{proposition}{Proposition}
\theoremstyle{remark}
\newtheorem*{remark}{Remark}
\theoremstyle{plain}
\newtheorem{corollary}{Corollary}

\begin{document}
	
	\begin{frontmatter}
		
		\title{Hybrid Reconstruction of Admissible Spline Spaces from Locally Modified Unclamped Patches for Isogeometric Analysis}
		
		\author[1]{Christopher Provatidis}
		\author[2]{Ioannis Dimitriou}
		
		\affiliation[1]{Professor of Mechanical Engineering, National Technical University of Athens}
		\affiliation[2]{PhD Candidate, National Technical University of Athens}
		
		\date{\today}
		
		\begin{abstract}
			The inherent global structure of NURBS knot vectors restricts localized spline modifications without propagating effects to adjacent regions, limiting the flexibility of adaptive Isogeometric workflows. To overcome this, we propose a decoupled reconstruction framework that temporarily decomposes a NURBS representation into independent local "Active Sections", enabling arbitrary local knot insertion, degree elevation, and basis modifications while preserving exact CAD geometry. However, these independent modifications inevitably violate inter-patch continuity, requiring a robust algebraic recovery of global admissibility. We introduce a novel reconstruction methodology that constructs a positive Hybrid Reconstruction Operator: anchor degrees of freedom are identified via QR pivoting, after which a sequence of linear programming problems with distance-based regularization generates strictly non-negative nullspace vectors, while a subsequent non-negative least-squares solve enforces a precise partition of unity. Crucially, to ensure computational efficiency and preserve locality, we implement a hierarchical pairwise condensation strategy that freezes columns already satisfying new interface constraints, confining the optimization exclusively to the active degrees of freedom at each merge step. The resulting hybrid basis is, by construction, strictly non-negative, spans the exact constrained nullspace, preserves partition of unity to machine precision, and reproduces the original geometry exactly. Numerical benchmarks, including a nonlinear diffusion problem on multi-patch curves with heterogeneous local polynomial degrees, confirm optimal convergence rates and demonstrate the framework's ability to seamlessly combine local geometric flexibility with globally consistent approximation spaces. By entirely decoupling local spline editing from the enforcement of continuity, this methodology provides a highly efficient, mathematically principled, and universally applicable tool for adaptive Isogeometric Analysis.
		\end{abstract}
		
		\begin{graphicalabstract}
		\end{graphicalabstract}
		
		\begin{highlights}
			\item A local hybrid B-spline framework enables exact geometry-preserving decomposition
			\item Null-space reconstruction guarantees exact recovery of the original geometry
			\item Active–frozen coordinate decomposition reduces reconstruction to local interface updates
			\item A minimal local knot vector of 2p+2 knots is proved for each active section
			\item The formulation supports efficient local refinement while preserving spline properties
		\end{highlights}
		
		
	\end{frontmatter}
	
	
\section{Introduction}
\label{sec:1}

Isogeometric Analysis (IGA)~\cite{hughes2005isogeometric} revolutionized computational mechanics by unifying Computer-Aided Design (CAD) geometry and numerical simulation through the use of NURBS basis functions. This unification enables exact geometric representation, high-order continuity, and superior per-degree-of-freedom accuracy compared to classical finite elements~\cite{cottrell2009isogeometric, zienkiewicz2005finite}. However, the same global tensor-product knot structure that provides these advantages also imposes a fundamental limitation: local spline modifications, such as knot insertion or degree elevation, inevitably propagate beyond the region of interest, affecting adjacent control variables and knot spans. Consequently, adaptive refinement or localized geometry editing in classical NURBS remains inherently non-local, restricting the flexibility required for efficient engineering analysis.


Over the past two decades, numerous spline technologies have been developed
to overcome the limitations of tensor-product NURBS. These include
hierarchical and truncated hierarchical B-splines (THB-splines),
T-splines, Locally Refined (LR) splines, U-splines, extraction-based
frameworks, and multi-patch coupling methodologies
\cite{vuong2011hierarchical,giannelli2012thb,
sederberg2003tsplines,bazilevs2006isogeometric,
dokken2013locally,herrema2018usplines,
borden2011isogeometric,popp2012dual}.
Collectively, these developments have significantly expanded the local
adaptivity and flexibility of Isogeometric Analysis. A detailed discussion
of these approaches is provided in Section~\ref{sec:3}.

Despite their diversity, all these approaches share a common paradigm: admissibility of the approximation space is maintained continuously throughout the construction or refinement process. Whether through hierarchical levels, T-junction constraints, or patch coupling, the continuity and compatibility of the spline space are enforced at every step. This paradigm, while mathematically sound, inherently couples local geometric manipulation with global admissibility enforcement. The present work originates from a fundamentally different question: Can we temporarily abandon global admissibility, perform arbitrary local spline manipulations independently, and only afterwards reconstruct a globally admissible space?

This question is motivated by a key observation: the exact geometry of a NURBS representation does not depend on global connectivity, but on the collection of local basis functions and their control data. If a local spline region can be isolated without altering the geometric mapping, then that region may be modified independently while preserving exact CAD geometry. However, such independent modifications inevitably destroy inter-patch continuity, rendering the resulting collection of local entities unsuitable for numerical analysis. We therefore propose a reconstruction-oriented framework that separates local spline manipulation from admissible-space construction. The methodology comprises three distinct stages: (i) decomposition of a NURBS representation into independent local entities called Active Sections---compact spline descriptions associated with a local knot support of size \(2p+2\), each retaining its own local knot vector, control points, and weights; (ii) independent local modification of each Active Section through arbitrary knot insertion, degree elevation, or basis modification, entirely independent of neighboring sections and preserving exact CAD geometry; and (iii) algebraic reconstruction of a globally admissible approximation space from the collection of locally modified Active Sections, through the assembly of continuity constraints into a global constraint operator \(\mathbf{C}_{tot}\) and the construction of a Hybrid Reconstruction Operator \(\mathbf{T}_{hybrid}\) that generates admissible basis functions while enforcing strict non-negativity and partition of unity.

The principal contribution of this work is not the introduction of a new spline technology, but rather a generic algebraic framework for reconstructing admissible spaces after arbitrary local spline manipulations. The novelty lies in: (i) the decoupling of local geometry editing from continuity enforcement, performed at different stages; (ii) an optimization-driven basis construction, where each column of the Hybrid Reconstruction Operator is obtained by solving a linear programming problem with distance-based regularization, promoting locality and strict non-negativity; (iii) an efficient hierarchical assembly strategy, where a pairwise condensation procedure identifies active and frozen columns at each merge step, confining optimization exclusively to degrees of freedom affected by new constraints, thereby ensuring computational scalability; and (iv) mathematical guarantees that the resulting basis is strictly non-negative, forms a partition of unity, spans the exact constrained nullspace, and preserves geometric exactness to machine precision.

The proposed framework complements, rather than replaces, existing spline technologies. Hierarchical methods, T-splines, and U-splines can be integrated within individual Active Sections prior to reconstruction, making the framework orthogonal to existing refinement strategies. In this sense, the methodology provides an additional layer of flexibility that can be superimposed on any spline representation capable of expressing continuity constraints algebraically. The present work focuses on the one-dimensional setting, where the reconstruction process can be examined with full clarity and controlled numerical experiments. Extension to multi-dimensional configurations introduces additional interface topology challenges, which are deferred to future work.

The remainder of the paper is organized as follows. Section 2 summarizes the necessary mathematical background. Section 3 reviews related spline technologies and coupling methodologies. Section 4 introduces the decomposition framework and the Active Section concept. Section 5 presents the reconstruction methodology and the construction of the Hybrid Reconstruction Operator. Section 6 presents numerical examples, including geometry preservation, continuity recovery, and a nonlinear diffusion benchmark. Section 7 discusses the properties, implications, limitations of the framework and outlines future research directions.

\section{Mathematical Background}
\label{sec:2}

We briefly summarize the essential concepts of B-spline and NURBS representations, local knot support, and spline spaces for analysis, which form the foundation of the proposed decomposition and reconstruction framework.

Let \(\Xi = \{\xi_1, \xi_2, \dots, \xi_{n+p+1}\}\) be a non-decreasing knot vector, where \(p\) denotes the polynomial degree and \(n\) the number of basis functions. The zeroth-order B-spline basis functions are defined by
\[
N_i^0(\xi) = \begin{cases}
1, & \xi_i \leq \xi < \xi_{i+1},\\
0, & \text{otherwise},
\end{cases}
\]
and higher-order basis functions are obtained recursively through the Cox--de Boor relation~\cite{piegl1997nurbs}:
\begin{equation}
N_i^p(\xi) = \frac{\xi - \xi_i}{\xi_{i+p} - \xi_i} N_i^{p-1}(\xi) + \frac{\xi_{i+p+1} - \xi}{\xi_{i+p+1} - \xi_{i+1}} N_{i+1}^{p-1}(\xi) \label{eq1} \,.
\end{equation}
B-spline basis functions, given by Eq.~\eqref{eq1}, possess several properties that are central to the present work: local support (each basis function vanishes outside a compact interval), non-negativity (\(N_i^p(\xi) \ge 0\)), partition of unity (\(\sum_i N_i^p(\xi) = 1\)), and controllable continuity through knot multiplicities. Specifically, a knot of multiplicity \(m\) yields continuity \(C^{p-m}\) at that location.

Non-Uniform Rational B-Splines (NURBS) extend B-splines through the introduction of positive weights \(w_i\). The rational basis functions are defined as
\begin{equation}
R_i(\xi) = \frac{N_i^p(\xi) w_i}{\sum_{j=1}^n N_j^p(\xi) w_j} \label{eq2} \,,
\end{equation}
and the geometric mapping is given by
\begin{equation}
\mathbf{x}(\xi) = \sum_{i=1}^n R_i(\xi) \mathbf{P}_i \label{eq3} \,,
\end{equation}
where \(\mathbf{P}_i\) are the control points. NURBS representations preserve the exact geometry of many engineering objects, including circles, conics, and free-form CAD models. This exactness plays a central role in the decomposition strategy proposed later in this work.




An important observation for the present framework concerns the local
support of B-spline basis functions. Each basis function of degree \(p\)
has support extending over at most \(p+1\) consecutive knot spans, and
exactly \(p+1\) basis functions are non-zero on the interior of any
nonzero knot span.

In the present framework, each individual nonzero knot span is associated
with an Active Section. The Active Section is completely characterized by
an inherited local knot vector containing
\[
2p+2
\]
knot entries: \(p\) knots to the left of the span, the two knots defining
the span itself, and \(p\) knots to the right. This local vector contains
all knot data required to evaluate the \(p+1\) basis functions that are
non-zero on the corresponding span.

Consequently, each Active Section provides a compact and self-contained
local spline description. Interactions between neighboring Active
Sections are introduced subsequently through algebraic continuity
constraints at their common interfaces. This compact local support forms
the basis of the decomposition and reconstruction procedures developed
in Sections~4 and~5.

In Isogeometric Analysis, the same basis functions employed for geometry representation are also used for field approximation~\cite{hughes2005isogeometric, cottrell2009isogeometric}. A generic approximation is written as
\begin{equation}
    u_h(\xi) = \sum_{A=1}^n R_A(\xi) a_A \label{eq4} \,,
\end{equation}
where \(a_A\) denote the control variables associated with the basis functions. The admissibility of the approximation space depends on continuity requirements imposed between neighboring spline regions. In classical IGA, these continuity conditions are embedded in the global spline construction. The present work focuses on reconstructing such admissible spaces after local spline entities have been manipulated independently, a process that temporarily destroys the original continuity relations.

\section{Related Work}
\label{sec:3}

The development of Isogeometric Analysis has been accompanied by extensive research on spline technologies for local refinement, adaptive discretization, and the construction of admissible approximation spaces. These developments have significantly extended the capabilities of classical tensor-product NURBS while preserving the geometric exactness that characterizes the isogeometric paradigm. Existing approaches may be broadly classified into local refinement technologies, extraction-based formulations, and multipatch coupling methods.

Hierarchical B-splines and their truncated variant (THB-splines) constitute one of the most widely adopted approaches for adaptive local refinement in Isogeometric Analysis~\cite{vuong2011hierarchical,giannelli2012thb}. Their central idea is the hierarchical activation of basis functions over selected regions of the computational domain, providing local refinement while preserving partition of unity, local support, and linear independence. T-splines~\cite{sederberg2003tsplines,bazilevs2006isogeometric,scott2011isogeometric} and Locally Refined (LR) splines~\cite{dokken2013locally} address the limitations of tensor-product refinement by introducing local topological modifications through T-junctions or locally refined mesh structures. Collectively, these spline technologies have considerably improved the flexibility of local refinement while maintaining admissible spline spaces throughout the refinement process.

A complementary line of research has focused on extraction-based formulations and spline constructions over unstructured topologies. Bézier extraction~\cite{borden2011isogeometric,scott2012analysis} provides an efficient algebraic framework for expressing spline basis functions in terms of Bernstein polynomials, thereby facilitating integration with finite element infrastructures. U-splines~\cite{herrema2018usplines,herrema2017adaptive} further extend this philosophy by enabling analysis on unstructured spline meshes while preserving compatibility with extraction-based implementations. More recent developments, including analysis-suitable spline constructions and multi-resolution formulations~\cite{takacs2025multiresolution,collin2016analysis,hughessangallitakacstoshniwal2021}, have broadened the range of admissible spline spaces available for geometric design and numerical analysis.

Another important research direction concerns the treatment of multipatch geometries. Mortar methods~\cite{wohlmuth2001mortar,popp2012dual}, Nitsche-type formulations~\cite{evans2009nitsche,apostolatos2014nitsche,hansbo2005nitsche}, weak coupling techniques~\cite{ruess2014weak,coox2017robust}, and domain decomposition methods~\cite{bernardi1993domain,quarteroni1999domain} provide powerful frameworks for enforcing continuity between independently parameterized spline patches. In addition, substantial effort has been devoted to the construction of smooth spline spaces over multipatch domains, including geometrically continuous parameterizations and analysis-suitable $C^1$ spline spaces~\cite{kapl2017isogeometric,collin2016analysis}. These methodologies have greatly expanded the applicability of Isogeometric Analysis to complex geometries while preserving high-order continuity.

Although the above technologies differ substantially in their mathematical construction and intended applications, they share a common principle: the admissibility of the approximation space is maintained throughout refinement, enrichment, or patch coupling. Local modifications are therefore performed within a spline space that remains globally admissible at every stage of the construction process.

The methodology proposed in the present work adopts a different perspective. Instead of constructing or refining a globally admissible spline space incrementally, it deliberately separates local spline manipulation from admissibility recovery. Independent local spline entities, referred to as Active Sections, are first modified without enforcing inter-section continuity. Global admissibility is subsequently recovered through an algebraic reconstruction procedure based on the assembly of continuity constraints and the construction of a positive Hybrid Reconstruction Operator. Consequently, the proposed framework is complementary to existing spline technologies rather than an alternative spline basis. Any spline representation capable of expressing continuity constraints algebraically may be employed within individual Active Sections prior to reconstruction. The principal contribution of the present work is therefore a generic algebraic methodology for reconstructing a globally admissible approximation space from independently modified local spline entities while preserving positivity, partition of unity, and exact geometric representation.

\section{Proposed Methodology}
\label{sec:4}

\subsection{Motivation}
\label{sec:4.1}

The central observation underlying the present work is that the exact geometry of a NURBS representation does not depend on the existence of a single globally connected spline description. Instead, the geometry is determined by the collection of local spline functions and their associated control data. Consequently, if a local spline region can be isolated without altering the geometric mapping, then that region may be manipulated independently while preserving the exact CAD representation. This observation motivates the introduction of the Active Section concept and the subsequent decomposition and reconstruction framework.

\subsection{Local Support and Definition of Active Sections}
\label{sec:4.2}

A fundamental property of B-spline basis functions is their compact support. Each basis function of degree \(p\) is non-zero only over \(p+1\) consecutive knot spans. Consequently, on any given knot span \([\xi_i, \xi_{i+1}]\), exactly \(p+1\) basis functions are non-zero. To evaluate these \(p+1\) basis functions at any point \(\xi \in [\xi_i, \xi_{i+1}]\), the Cox--de Boor recurrence requires a local knot vector that extends \(p\) knots to the left of \(\xi_i\) and \(p\) knots to the right of \(\xi_{i+1}\). Thus, the minimal local knot sequence that completely determines the basis functions on the span is
\begin{equation}
\{\xi_{i-p}, \xi_{i-p+1}, \dots, \xi_i, \xi_{i+1}, \dots, \xi_{i+p+1}\} \,,
\end{equation}
which contains exactly \(2p+2\) knots.

This observation is illustrated in Figure~1 for \(p=2\). The figure shows a clamped global knot vector \(\Xi = \{0,0,0,0.25,0.5,0.75,1,1,1\}\), with dashed vertical lines across the full interval. The chosen knot span \([\xi_i, \xi_{i+1}] = [0.25, 0.5]\) is highlighted in red. On this span, exactly \(p+1=3\) basis functions are non-zero (coloured curves). The local knot vector required for their evaluation consists only of the \(2p+2=6\) knots immediately surrounding the span, which are marked in blue at the bottom of the figure: \(\{0,0,0.25,0.5,0.75,1\}\) (i.e., \(p=2\) knots to the left, the two span endpoints, and \(p=2\) knots to the right). All other knots of the global vector are irrelevant for the local shape of these active functions.

\begin{figure}[htbp]
    \centering
    
     \includegraphics[width=1.0\linewidth]{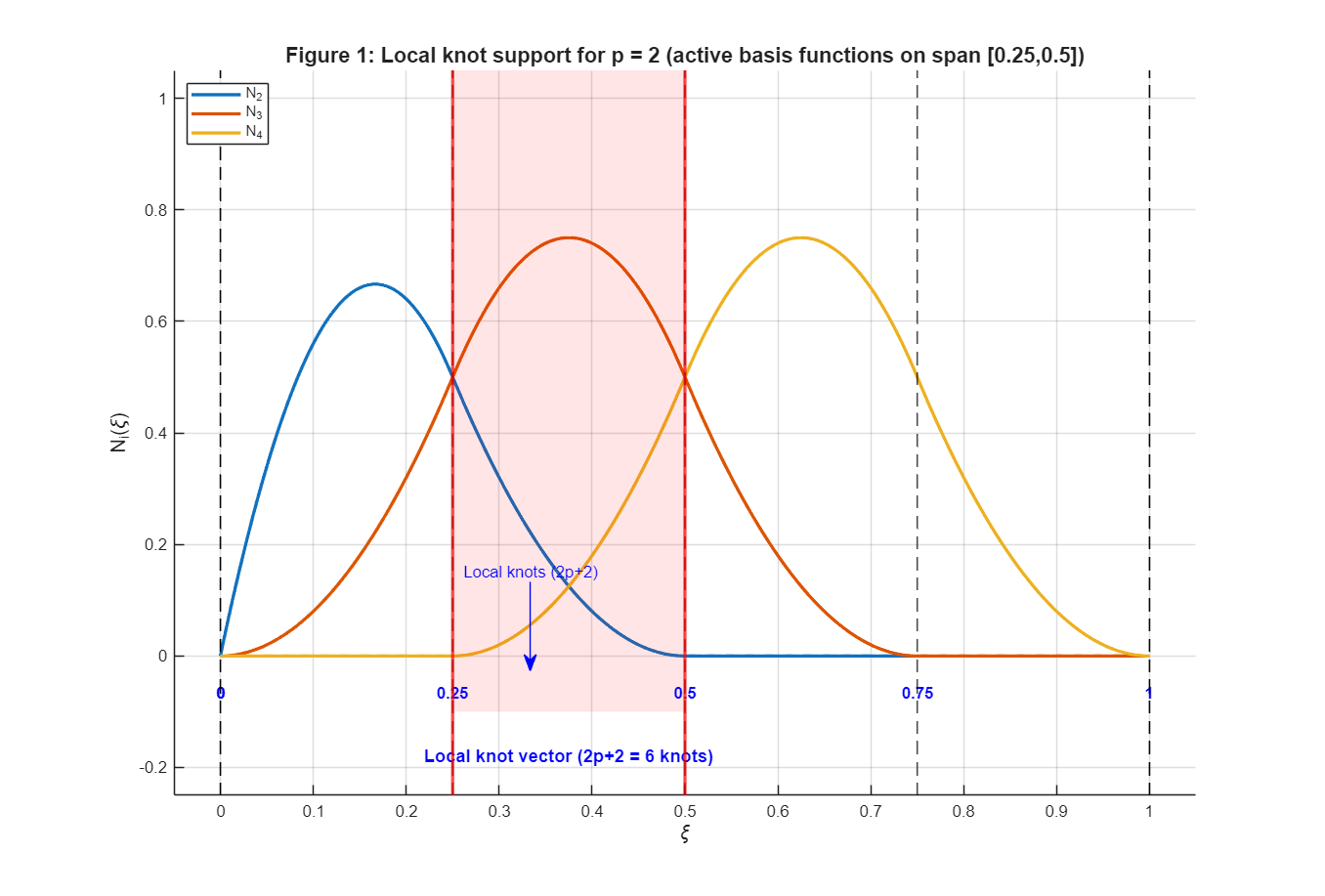} 
    \caption{
    Local knot support for B-spline basis functions of degree $p=2$.
    The global clamped knot vector is $\Xi = \{0,0,0,0.25,0.5,0.75,1,1,1\}$, indicated by dashed vertical lines across the entire interval.
    The chosen knot span $[\xi_i, \xi_{i+1}] = [0.25, 0.5]$ is highlighted in red.
    On this span, exactly $p+1 = 3$ basis functions (coloured curves) are non-zero.
    Their evaluation requires only the $2p+2 = 6$ local knots immediately surrounding the span, which are marked in blue at the bottom of the figure: $\{0,0,0.25,0.5,0.75,1\}$ (i.e., $p=2$ knots to the left, the two span endpoints, and $p=2$ knots to the right).
    All other knots of the global vector are irrelevant for the local shape of the active basis functions on this span.
    This compact local neighborhood directly motivates the Active Section concept introduced below.
    } 
    \label{fig0}
\end{figure}

The figure clearly demonstrates that the evaluation on a given span depends only on this local sub‑sequence. This local window has a crucial implication: any modification confined to the interior of this window—provided the outermost \(p\) knots on each side remain unchanged—does not affect the basis functions outside the window. The first \(p\) and last \(p\) knots therefore act as buffers that isolate the internal spline data from the surrounding geometry. We exploit this property to define the Active Section.

\begin{definition}
For a spline space of degree \(p\), the \emph{Active Section} associated
with a nonzero knot span \([\xi_i,\xi_{i+1}]\) is the smallest local
spline entity that completely determines all basis functions non-zero
on that span. Specifically, the \(p+1\) basis functions that are non-zero
on \([\xi_i,\xi_{i+1}]\) can be evaluated using only the local knot vector
\begin{equation}
\Xi_{\mathrm{AS}}^{(i)}
=
\{\xi_{i-p},\xi_{i-p+1},\dots,\xi_{i+p+1}\},
\end{equation}
which contains \(2p+2\) knot entries. The first \(p\) and the last \(p\)
entries form the left and right buffer regions, respectively, whereas
the two central entries \(\xi_i\) and \(\xi_{i+1}\) define the knot span
associated with the Active Section. Local spline modifications may
subsequently be introduced within this span while the two outer buffer
regions are retained.
\end{definition}

This is illustrated in Figure~\ref{fig0}, where the six blue knots fully determine the three active basis functions on the highlighted span.

The extraction of these Active Sections from a parent NURBS representation, and their interpretation as independent unclamped NURBS elements, is described in the following subsection.

The definition of the Active Section introduced above is not an
algorithmic assumption but a direct consequence of the local support
properties of B-spline basis functions. In particular, the local knot
vector containing exactly \(2p+2\) knot entries is the unique minimal
knot sequence required to evaluate all basis functions that are non-zero
over the associated knot span. This fundamental property constitutes the
mathematical basis of the proposed decomposition framework. A complete
proof of the minimality of the \(2p+2\) local knot vector is provided in
Appendix~\ref{app:B}.

\subsection{Extraction of Local Unclamped NURBS Elements}
\label{sec:4.3}

The decomposition procedure starts from an initial NURBS representation that has already undergone the required global \(h\)- or \(p\)-refinement, if necessary (to start the proposed method, the minimum number of elements should be two). The refined parent geometry is assumed to contain at least two knot spans, so that local interface regions can be identified. For a parent spline of degree \(p\), each local entity is extracted using a compact knot support containing \(2p+2\) knot entries. These local entities are referred to as \emph{unclamped NURBS elements}. Each extracted element contains its own local knot vector, local weights, and associated control points inherited from the parent NURBS representation.

The extraction process may be written schematically as
\begin{equation}
\mathcal{N}_{\text{parent}} \longrightarrow \{\mathcal{N}_1^{\text{loc}}, \mathcal{N}_2^{\text{loc}}, \dots, \mathcal{N}_{n_e}^{\text{loc}}\},
\end{equation}
where \(\mathcal{N}_{\text{parent}}\) denotes the refined parent NURBS representation and \(\mathcal{N}_e^{\text{loc}}\) denotes the local unclamped NURBS representation associated with element \(e\). In the present implementation, each local entity is rebuilt as an independent NURBS object using its extracted knot vector, control points, and weights. The result is a structured collection of local NURBS objects,
\begin{equation}
\text{geom} = \{\text{geom}(1), \text{geom}(2), \dots, \text{geom}(n_e)\},
\end{equation}
where each entry represents one unclamped local NURBS element.

This step changes only the representation of the geometry. The local NURBS elements are extracted from the parent representation together with their corresponding control data; therefore the original geometric description is preserved. Subsequent local operations, such as knot insertion or degree elevation, are then performed directly on the individual entries of this local NURBS structure.

\subsection{Geometry Preservation}
\label{sec:4.4}

The decomposition process modifies neither the geometric mapping nor the physical location of any point on the curve. Only the representation changes. Consequently,
\begin{equation}
\mathbf{x}_{\text{original}}(\xi) = \mathbf{x}_{\text{decomposed}}(\xi) \qquad \forall \xi.
\end{equation}
The decomposition therefore constitutes an exact reparameterization of the spline description rather than a geometric approximation. This property is fundamental because it permits local spline operations to be performed without introducing geometric errors.

\subsection{Local Spline Manipulations}
\label{sec:4.5}

Once the NURBS representation has been decomposed into independent Active Sections, local spline operations may be performed directly on the associated local knot vectors. Contrary to classical NURBS refinement procedures, which are typically applied within the context of a global knot vector, the proposed framework treats each Active Section as an autonomous spline entity. Consequently, refinement operations are confined to the selected local region and do not require modification of neighboring sections.

Let
\begin{equation}
\Xi_e =
\left\{
\xi_1^{(e)},\xi_2^{(e)},\ldots,\xi_{n_e}^{(e)}
\right\},
\label{eq:local_knot_vector}
\end{equation}

denote the current local knot vector associated with Active Section
$e$, where initially $n_e=2p+2$.

\subsubsection{Local Knot Insertion}

To preserve the information required for subsequent continuity
reconstruction, the local knot vector is partitioned according to
its positional indices as
\begin{equation}
\Xi_e =
\left\{
\underbrace{
\xi_1^{(e)},\ldots,\xi_p^{(e)}
}_{\text{left buffer}},
\;
\underbrace{
\xi_{p+1}^{(e)},\ldots,\xi_{n_e-p}^{(e)}
}_{\text{locally modifiable region}},
\;
\underbrace{
\xi_{n_e-p+1}^{(e)},\ldots,\xi_{n_e}^{(e)}
}_{\text{right buffer}}
\right\}.
\label{eq:local_knot_partition}
\end{equation}

Thus, the first $p$ and the last $p$ knot entries are retained
unchanged, whereas local spline modifications are restricted to the
central block. In positional form, this partition is
\[
1:p
\;\big|\;
p+1:\mathrm{end}-p
\;\big|\;
\mathrm{end}-p+1:\mathrm{end}.
\]

For the initially extracted Active Section, $n_e=2p+2$, and hence
the central block contains exactly the two endpoints of the original
active knot span. After local knot insertion, $n_e$ may increase,
while the first $p$ and last $p$ entries remain unchanged.

The admissible interval for local knot insertion is therefore
\begin{equation}
I^{(e)}
=
\left[
\xi_{p+1}^{(e)},
\xi_{n_e-p}^{(e)}
\right],
\label{eq:local_interval}
\end{equation}
which defines the locally modifiable region within the Active Section.
Local knot insertion is performed only for knots satisfying
\begin{equation}
\tilde{\xi}
\in
\left(
\xi_{p+1}^{(e)},
\xi_{n_e-p}^{(e)}
\right).
\label{eq:local_knot_insertion}
\end{equation}

The insertion is carried out using the standard NURBS knot insertion
algorithm~\cite{piegl1997nurbs}. Since knot insertion is an exact
geometric operation,
\begin{equation}
\mathbf{x}_{\mathrm{before}}(\xi)
=
\mathbf{x}_{\mathrm{after}}(\xi),
\qquad \forall \xi .
\label{eq:knot_insertion_geometry}
\end{equation}

Unlike clamped spline representations, the local knot vectors employed
in the present work are generally unclamped. Consequently, refinement
cannot be interpreted in terms of preserving endpoint multiplicities
of order $p+1$. Instead, the outer $p$ knot entries on each side are
retained as buffer regions, since they contain the spline information
required for subsequent admissibility reconstruction.

\subsubsection{Local Degree Modification}
\label{sec:4.5.2}

Direct degree elevation of the extracted unclamped NURBS elements has not been implemented in the present work. Instead, local degree modification is achieved through a re-extraction procedure based on globally elevated NURBS representations. Consider an initial NURBS geometry of degree \(p\) together with its corresponding decomposition into local unclamped elements. If a higher local degree is required, the parent NURBS representation is first elevated globally to degree
\begin{equation}
p^* = p + r,
\end{equation}
where \(r\) denotes the degree increment. The elevated geometry is subsequently decomposed using the same Active Section extraction procedure. Since degree elevation preserves the exact geometry~\cite{piegl1997nurbs},
\begin{equation}
\mathbf{x}_p(\xi) = \mathbf{x}_{p^*}(\xi), \qquad \forall \xi.
\end{equation}
Local degree modification is then achieved by selecting the desired elements from the higher-degree decomposition and inserting them into the original collection of Active Sections. In this manner, neighboring sections may possess different polynomial degrees while continuing to represent the same underlying geometry.

Note that p-refinement does not increase the number of elements. For example, the global knot vector $\Xi_{1,glob}=[0,0,0,0.5,1,1,1]$ (with $p=2$) has the same elements with $\Xi_{2,glob}=[0,0,0,0,0.5,0.5,1,1,1,1]$ (with $p=3$ plus an inserted knot at $\xi=0.5$). If the first element (with $p=2$) is fully described by the local knot vector $\Xi_{1,loc}=[0,0,0,0.5,1,1]$ (of size $2p+2$), then the first element (with $p=3$) will be fully described by $\Xi_{i,loc}'=[0,0,0,0,0.5,0.5,1,1]$ in the p-refined state. Therefore, we substitute with the $p+1$ element and have the same geometry, where the first element will be  $p=3$ and the second $p=2$.

\subsection{Consequences of Local Manipulation}
\label{sec:4.6}

The decomposition and refinement procedures preserve the exact geometry of the original NURBS representation. However, they do not preserve the continuity relationships that existed in the original spline space. After independent modifications have been performed, neighboring Active Sections generally possess:
\begin{itemize}
    \item different knot vectors,
    \item different polynomial degrees,
    \item different basis representations,
    \item different continuity characteristics.
\end{itemize}
Consequently, the collection of modified Active Sections no longer defines a globally admissible spline space. The geometry remains exact, but the approximation space required for numerical analysis has been lost.

This distinction is fundamental. From a geometric perspective, the decomposition process is complete. The exact CAD representation has been preserved throughout all local operations. From an analysis perspective, however, additional work is required. Continuity conditions must be re-established and a new admissible approximation space must be constructed.

The remainder of this work is devoted to this reconstruction problem, which is addressed in Section 5 through the construction of a Hybrid Reconstruction Operator based on interface continuity constraints and nullspace optimization.

\subsection{From Geometry to Analysis}
\label{sec:4.7}

The decomposition framework separates two traditionally coupled objectives: (i) geometric manipulation, and (ii) admissible-space construction. The first objective is achieved through the use of independent Active Sections and local spline operations. The second objective requires the construction of an admissible approximation space satisfying the continuity requirements of the target analysis problem. Several reconstruction strategies were investigated during the development of the present work. Direct coupling approaches were found to be restrictive when local knot vectors and polynomial degrees differed significantly between neighboring Active Sections. This observation motivated the investigation of algebraic reconstruction procedures based on continuity constraints and null-space operators, which are presented in the following section.


\section{Local‑to‑Global Reconstruction}
\label{sec:5}

\subsection{The Reconstruction Problem}
\label{sec:5.1}

The decomposition strategy introduced in the previous section allows a NURBS representation to be expressed as a collection of independent Active Sections. Each section possesses its own local knot vector and may undergo local spline operations such as knot insertion, degree elevation, or basis modification without affecting neighboring regions. From a geometric perspective, these operations preserve the exact CAD representation. However, the local modifications performed within individual Active Sections generally destroy the continuity relationships originally present in the global spline space. Neighboring sections may possess different polynomial degrees, different local knot vectors, or different basis representations. As a consequence, the collection of locally modified sections no longer forms an admissible approximation space suitable for numerical analysis.

The objective of the reconstruction stage is therefore to generate a new approximation space that satisfies the desired continuity requirements while retaining all local modifications introduced during the decomposition phase. Let
\begin{equation}
\mathbf{N}_{\text{loc}} = \begin{bmatrix} N_1 & N_2 & \cdots & N_{n_{\text{loc}}} \end{bmatrix} \,,
\end{equation}
denote the collection of basis functions associated with all locally modified Active Sections. The reconstruction problem may then be stated as follows:

\begin{quote}
\textit{Given a collection of locally modified spline entities, construct an admissible approximation space that satisfies prescribed continuity requirements while preserving the geometric and approximation properties introduced during local refinement.}
\end{quote}

The remainder of this section develops a reconstruction framework that addresses this problem through the assembly of continuity constraints and the construction of a Hybrid Reconstruction Operator.

\subsection{Interface Continuity Constraints}
\label{sec:5.2}

The reconstruction procedure is based on the observation that admissibility is governed entirely by the continuity relationships between neighboring Active Sections. Once the local spline entities have been modified independently, these continuity relations are generally no longer satisfied and must be re-established before numerical analysis can be performed.

Consider two neighboring Active Sections, denoted by \(\Omega_A\) and \(\Omega_B\), sharing a common interface. Let
\begin{equation}
u_A(\xi) = \sum_{i=1}^{n_A} N_i^A(\xi) d_i^A
\qquad\text{and}\qquad
u_B(\xi) = \sum_{j=1}^{n_B} N_j^B(\xi) d_j^B \,,
\end{equation}
represent the local approximations on the two sections. To recover an admissible approximation space, continuity conditions must be enforced at the interface. The specific form of these conditions depends on the continuity requirements of the target problem. For \(C^0\) continuity, the reconstructed field must satisfy
\begin{equation}
u_A(\xi_I) = u_B(\xi_I) \,,
\end{equation}
where \(\xi_I\) denotes the interface location. For \(C^1\) continuity, the first derivatives must additionally satisfy
\begin{equation}
\frac{du_A}{d\xi}(\xi_I) = \frac{du_B}{d\xi}(\xi_I) \,.
\end{equation}
More generally, for a prescribed continuity order \(C^k\), the following conditions are imposed:
\begin{equation}
\frac{d^m u_A}{d\xi^m}(\xi_I) = \frac{d^m u_B}{d\xi^m}(\xi_I), \qquad m = 0, 1, \dots, k \,.
\end{equation}

Substitution of the local spline approximations into these relations produces a collection of linear equations involving the local control variables associated with the neighboring Active Sections. These equations define the admissibility conditions that must be satisfied by the reconstructed approximation space. Each interface therefore contributes a set of algebraic constraints linking the local degrees of freedom of adjacent sections. An important characteristic of these constraints is their locality: each continuity condition involves only the basis functions participating in the interaction across the corresponding interface.

\subsection{Pairwise Admissible Coupling}
\label{sec:5.3}

The reconstruction framework is built upon a sequence of pairwise
coupling operations between neighboring Active Sections. Rather than
constructing the admissible approximation space globally, continuity is
introduced progressively through local interface interactions.

Consider two neighboring Active Sections with local basis
representations
\begin{equation}
\mathbf{N}_A=
\begin{bmatrix}
N_1^A&\cdots&N_{n_A}^A
\end{bmatrix},
\qquad
\mathbf{N}_B=
\begin{bmatrix}
N_1^B&\cdots&N_{n_B}^B
\end{bmatrix}.
\end{equation}

The continuity conditions described in the previous subsection generate
a local constraint operator
\(\mathbf C_{AB}\)
relating the local control variables of the two neighboring Active
Sections.

The admissible coefficient space associated with the pair is obtained by
enforcing

\begin{equation}
\mathbf C_{AB}\mathbf d_{AB}=\mathbf0,
\end{equation}

where
\(\mathbf d_{AB}\)
collects the local degrees of freedom of the two Active Sections.

The resulting admissible approximation is represented by a locally
reconstructed block satisfying the prescribed continuity conditions
across the interface. This reconstructed block is subsequently coupled
with the next neighboring Active Section.

The reconstruction therefore proceeds recursively through a sequence of
pairwise admissible couplings. For a collection of Active Sections

\begin{equation}
\{\mathcal N_1,\mathcal N_2,\ldots,\mathcal N_m\},
\end{equation}

the recursive reconstruction is expressed as

\begin{equation}
\mathcal B_{12}
=
\mathcal R(\mathcal N_1,\mathcal N_2),
\qquad
\mathcal B_{123}
=
\mathcal R(\mathcal B_{12},\mathcal N_3),
\qquad
\dots,
\qquad
\mathcal B_{1\ldots m}
=
\mathcal R(\mathcal B_{1\ldots m-1},\mathcal N_m).
\end{equation}

where \(m\) denotes the number of Active Sections
(\(m-1\) interfaces and \(m-1\) pairwise reconstruction steps).

The final Hybrid Spline Space is therefore obtained through a sequence
of local admissible reconstructions rather than by solving a single
global reconstruction problem.

Algorithmic details together with the mathematical proof of the
recursive construction are presented in
Appendix~\ref{app:A}.

To illustrate the recursive procedure, we first consider the simplest
case of \(C^0\) continuity between two neighboring Active Sections. In
this setting a single interface constraint is imposed, allowing the
complete pairwise reconstruction algorithm to be presented before its
extension to higher-order continuity (see Example~1).

\subsection{Global Constraint Representation}
\label{sec:5.4}

Although the reconstruction process is naturally described through successive pairwise admissible couplings, it is often convenient to represent all continuity relations within a single algebraic framework. Each pairwise coupling operation contributes a set of local continuity constraints involving only the degrees of freedom associated with the corresponding interface. Let
\[
\mathbf{C}_{AB}, \quad
\mathbf{C}_{BC}, \quad
\mathbf{C}_{CD}, \quad
\dots
\]
denote the local constraint operators generated during the reconstruction process. The complete set of admissibility conditions is assembled into a global constraint operator
\(\mathbf{C}_{\mathrm{tot}}\), which collects all interface continuity constraints associated with the reconstructed spline space.

Let
\[
m_{\mathrm{loc}}
=
\sum_{e=1}^{N_{\mathrm{AS}}} n_e
\]
denote the total number of local degrees of freedom after decomposition
and local refinement, where \(N_{\mathrm{AS}}\) is the number of Active
Sections and \(n_e\) is the number of local basis functions associated
with Active Section \(e\).

Define the vector of all local control variables as
\begin{equation}
\mathbf d
=
\begin{bmatrix}
d_1 &
d_2 &
\cdots &
d_{m_{\mathrm{loc}}}
\end{bmatrix}^{T}.
\end{equation}

The admissibility conditions may then be written compactly as
\begin{equation}
\mathbf C_{\mathrm{tot}}\mathbf d=\mathbf0.
\end{equation}

Each row of \(\mathbf C_{\mathrm{tot}}\) represents one continuity
constraint, whereas each column corresponds to one local degree of
freedom. Because every pairwise coupling involves only neighboring
Active Sections, each continuity condition affects only a small subset
of the local variables. Consequently, the global constraint operator
remains sparse even for large reconstructed spline spaces.

The admissible approximation space is therefore defined as the null space
of the global constraint operator,
\begin{equation}
\mathcal V_{\mathrm{adm}}
=
\left\{
\mathbf d\in\mathbb R^{m_{\mathrm{loc}}}
:
\mathbf C_{\mathrm{tot}}\mathbf d=\mathbf0
\right\}.
\end{equation}

The following subsections introduce the Hybrid Reconstruction Operator,
which constructs a positive basis spanning the admissible approximation
space while preserving partition of unity and the exact null space of
\(\mathbf C_{\mathrm{tot}}\).

\subsection{Construction of the Hybrid Reconstruction Operator}
\label{sec:5.5}

Let $m_{\mathrm{loc}}$ denote the total number of local degrees of
freedom after decomposition and local refinement (associated with $m$ active sections), and let

\begin{equation}
    r=\operatorname{rank}(C_{\mathrm{tot}}) \,.
\end{equation}
The number of independent degrees of freedom of the reconstructed space is
\begin{equation}
    n_{\mathrm{hyb}} = m_{\mathrm{loc}}-r \,.
\end{equation}
Although any basis of ker(Ctot), for example the orthonormal basis returned by the MATLAB \texttt{null} function, spans the admissible space, such bases generally do not preserve positivity, partition of unity, or local support. The objective of the proposed reconstruction operator is therefore to construct an admissible basis possessing these geometric and numerical properties.

\subsubsection{Anchor Selection}
\label{sec:5.5.1}

The first step consists of identifying a set of independent degrees of freedom capable of generating the admissible space. A pivoted QR factorization of the global constraint matrix is computed,

\begin{equation}
    C_{\mathrm{tot}}\Pi
    =
    Q
    \left[
    R_1\;
    R_2
    \right] \,,
\end{equation}
where the permutation matrix $\Pi$ reorders the columns so that the first
$r=\operatorname{rank}(C_{\mathrm{tot}})$ columns are the pivot columns
associated with the dependent variables, whereas the remaining
$n_{\mathrm{hyb}}$ columns correspond to the free variables.

\subsubsection{Positive Admissible Generators}
\label{sec:5.5.2}

For each anchor location $a_j$, we seek a vector
$\mathbf{t}_j \in \mathbb{R}^{m_{\mathrm{loc}}}$
representing a positive admissible generator associated with that anchor.
The generator is obtained by solving the following linear programming problem:
\begin{equation}
    \begin{aligned}
    \min_{\mathbf{t}_j} \quad & \mathbf{f}_j^T \mathbf{t}_j, \\
    \text{s.t.} \quad & \mathbf{C}_{\text{tot}} \mathbf{t}_j = \mathbf{0}, \\
    & (\mathbf{t}_j)_{a_j} = 1, \\
    & 0 \leq \mathbf{t}_j \leq 1 \,.
    \end{aligned}
\end{equation}
The first constraint guarantees admissibility, the second fixes the anchor location, while the bound constraints enforce positivity. The objective vector \(\mathbf{f}_j\) is selected such that coefficients located farther from the anchor are mildly penalized. In the present work, we employ
\begin{equation}
    (\mathbf{f}_j)_i = 1 + \alpha (i - a_j)^2,
\end{equation}
where \(\alpha\) is a small positive constant (e.g., \(\alpha = 0.01\)). Consequently, the resulting generators remain as localized as possible while satisfying the admissibility requirements.

The admissible generators are assembled into the matrix
\begin{equation}
    \mathbf{C}_{\text{pos}} = \begin{bmatrix} \mathbf{t}_1 & \mathbf{t}_2 & \cdots & \mathbf{t}_{n_{\text{hyb}}} \end{bmatrix}.
\end{equation}
By construction,
\begin{equation}
    \mathbf{C}_{\text{tot}} \mathbf{C}_{\text{pos}} = \mathbf{0}.
\end{equation}

\subsubsection{Partition of Unity Recovery}

Although the admissible generators satisfy the continuity constraints, they do not necessarily satisfy partition of unity. To recover this property, a non-negative scaling vector
\begin{equation}
    \boldsymbol{\gamma} = \begin{bmatrix} \gamma_1 & \gamma_2 & \cdots & \gamma_{n_{\text{hyb}}} \end{bmatrix}^T
\end{equation}
is determined from the constrained least-squares problem
\begin{equation}
    \mathbf{C}_{\text{pos}} \boldsymbol{\gamma} \approx \mathbf{1},
\end{equation}
where \(\mathbf{1} \in \mathbb{R}^{m}\) is the vector of all ones. Specifically, we solve the non-negative least squares (NNLS) problem
\begin{equation}
    \min_{\boldsymbol{\gamma} \ge 0} \left\| \mathbf{C}_{\text{pos}} \boldsymbol{\gamma} - \mathbf{1} \right\|_2.
\end{equation}
The final Hybrid Reconstruction Operator is then defined as
\begin{equation}
    \mathbf{T}_{\text{hybrid}} = \mathbf{C}_{\text{pos}} \operatorname{diag}(\boldsymbol{\gamma}).
\end{equation}
Additional balancing operations may be introduced to improve the uniformity of the reconstructed basis while preserving positivity and admissibility. In the present implementation, we apply a gentle column scaling to balance the peak heights of the basis functions, followed by a second NNLS solve to restore partition of unity.

\subsubsection{Properties of the Reconstruction Operator}
\label{sec:5.5.4}

The resulting operator satisfies three fundamental properties:
\begin{align}
\mathbf{C}_{\text{tot}} \mathbf{T}_{\text{hybrid}} &= \mathbf{0}, &&\text{(admissibility)}, \\
\mathbf{T}_{\text{hybrid}} &\ge 0, &&\text{(non-negativity)}, \\
\mathbf{T}_{\text{hybrid}} \mathbf{1} &\approx \mathbf{1}, &&\text{(partition of unity)}.
\end{align}
The operator therefore combines admissibility, positivity, and partition-of-unity preservation within a single reconstruction framework. For this reason, the resulting basis is referred to as a \emph{Hybrid Spline Space}.

\subsection{Efficient Hierarchical Assembly via Active/Frozen Columns}
\label{sec:5.6}

A critical aspect of the proposed framework is computational efficiency. Although the global constraint operator \(\mathbf{C}_{\text{tot}}\) can be assembled directly, solving the full linear programming problem for all anchors simultaneously becomes expensive for problems with many Active Sections. To address this, we implement a hierarchical condensation strategy that processes patches sequentially, freezing columns that already satisfy new interface constraints and confining optimization exclusively to the active degrees of freedom at each merge step.

Suppose we have already constructed a valid basis \(\mathbf{T}_{\text{cur}} \in \mathbb{R}^{N_{\text{cur}} \times d_{\text{cur}}}\) for the first \(k-1\) patches. To incorporate patch \(k\), we proceed as follows:

\paragraph{Augmentation.}
Extend the current basis with the identity of the new patch:
\begin{equation}
    \mathbf{T}_{\text{aug}} = \begin{bmatrix}
    \mathbf{T}_{\text{cur}} & \mathbf{0} \\
    \mathbf{0} & \mathbf{I}_{n_k}
    \end{bmatrix}
    \in \mathbb{R}^{(N_{\text{cur}} + n_k) \times (d_{\text{cur}} + n_k)} \,.
\end{equation}

\paragraph{Constraint Projection.}
Let \(\mathbf{C}_{\text{int}} \in \mathbb{R}^{c \times (N_{\text{cur}} + n_k)}\) be the matrix representing the continuity conditions across the interface \((k-1 | k)\). We project these constraints onto the augmented basis:
\begin{equation}
    \mathbf{C}_{\text{red}} = \mathbf{C}_{\text{int}} \, \mathbf{T}_{\text{aug}} \,.
\end{equation}

\paragraph{Active/Frozen Split.}
A column \(i\) of \(\mathbf{T}_{\text{aug}}\) is considered \emph{frozen} if the corresponding column of \(\mathbf{C}_{\text{red}}\) is zero (within a tolerance), meaning it already satisfies the new interface condition. Otherwise, it is \emph{active}. Let \(\mathcal{F}\) and \(\mathcal{A}\) denote the sets of frozen and active column indices, respectively. We have
\begin{equation}
    \|\mathbf{C}_{\text{red}}(:, i)\| = 0 \quad \forall i \in \mathcal{F},
    \qquad
    \|\mathbf{C}_{\text{red}}(:, i)\| > 0 \quad \forall i \in \mathcal{A} \,.
\end{equation}

\paragraph{Local Reduction.}
We apply the positive basis construction (Sections 5.5.2–5.5.3) exclusively to the active subsystem \(\mathbf{C}_{\text{red}}(:, \mathcal{A})\), yielding \(\mathbf{T}_{\text{active}}\). We then build the global reduction matrix:
\begin{equation}
    \mathbf{T}_{\text{red}} = 
    \begin{bmatrix}
    \mathbf{I}_{|\mathcal{F}|} & \mathbf{0} \\
    \mathbf{0} & \mathbf{T}_{\text{active}}
    \end{bmatrix} \,,
\end{equation}
where the frozen columns are left unchanged.

\paragraph{Update.}
The basis for the extended patch system becomes
\begin{equation}
    \mathbf{T}_{\text{cur}} \leftarrow \mathbf{T}_{\text{aug}} \, \mathbf{T}_{\text{red}} \,.
\end{equation}
We repeat this process sequentially for patches \(3, 4, \dots, M\). The final matrix \(\mathbf{T}_{\text{cur}}\) is the desired global hybrid basis \(\mathbf{T}_{\text{hybrid}}\). This hierarchical strategy ensures that the optimization burden is confined to the active degrees of freedom at each merge step, making the framework computationally scalable for problems with many patches.

\subsection{Summary of the Reconstruction Algorithm}
\label{sec:5.7}

The complete construction of the Hybrid Reconstruction Operator is
summarized in Algorithm~\ref{alg:hybrid_reconstruction}. The procedure
follows the sequential pairwise reconstruction described above, while
the corresponding mathematical justification is provided in
Appendix~\ref{app:A}.

\begin{center}
\captionsetup{type=algorithm}
\captionof{algorithm}{Sequential active--frozen construction of the
Hybrid Reconstruction Operator}
\label{alg:hybrid_reconstruction}
\end{center}

\begin{algorithmic}[1]

\Require Locally modified Active Sections
$\{\mathcal{A}_1,\ldots,\mathcal{A}_m\}$,
their local basis functions, prescribed interface continuity conditions,
and numerical tolerance $\varepsilon$

\Ensure Hybrid Reconstruction Operator $T_{\mathrm{hybrid}}$

\State Initialize
$T_1 \gets I_{n_1}$,
where $n_1$ is the number of local basis functions of $\mathcal{A}_1$

\For{$k=2,\ldots,m$}

    \State Augment the current transformation with Active Section
    $\mathcal{A}_k$:
    \Statex \hspace{\algorithmicindent}
    $T_{\mathrm{aug},k}
    \gets
    \operatorname{blkdiag}(T_{k-1},I_{n_k})$

    \State Assemble the interface constraint matrix $C_k$
    associated with the prescribed continuity conditions between
    $\mathcal{A}_1,\ldots,\mathcal{A}_{k-1}$ and $\mathcal{A}_k$

    \State Project the new interface constraints:
    \Statex \hspace{\algorithmicindent}
    $\bar{C}_k
    \gets
    C_kT_{\mathrm{aug},k}$

    \State Identify the frozen coordinates:
    \Statex \hspace{\algorithmicindent}
    $\mathcal{F}_k
    \gets
    \left\{
    j:
    \|\bar{C}_k(:,j)\|_2
    \leq\varepsilon
    \right\}$

    \State Let $\mathcal{A}_k^{\mathrm{act}}$
    denote the complementary set of active coordinates

    \If{$\mathcal{A}_k^{\mathrm{act}}\neq\varnothing$}

        \State Form the active constraint matrix
        $C_{\mathrm{act},k}$
        from the active columns of $\bar{C}_k$

        \State Determine independent anchor coordinates using
        column-pivoted QR factorization of $C_{\mathrm{act},k}$

        \For{each selected anchor $a_j$}

            \State Compute a non-negative admissible generator $t_j$
            by solving
            \Statex \hspace{\algorithmicindent}
            $\displaystyle
            \min_{t_j}\; f_j^Tt_j
            $
            \Statex \hspace{\algorithmicindent}
            subject to
            $\;
            C_{\mathrm{act},k}t_j=0,\quad
            (t_j)_{a_j}=1,\quad
            0\leq t_j\leq1
            $

        \EndFor

        \State Collect the positive generators:
        \Statex \hspace{\algorithmicindent}
        $T_{\mathrm{pos},k}
        \gets
        [\,t_1\;t_2\;\cdots\;t_q\,]$

        \State Determine the non-negative scaling vector:
        \Statex \hspace{\algorithmicindent}
        $\displaystyle
        \gamma_k
        \gets
        \arg\min_{\gamma\geq0}
        \|T_{\mathrm{pos},k}\gamma-\mathbf{1}\|_2
        $

        \State Define the active reconstruction matrix:
        \Statex \hspace{\algorithmicindent}
        $T_{\mathrm{act},k}
        \gets
        T_{\mathrm{pos},k}\operatorname{diag}(\gamma_k)$

        \State Construct $T_{\mathrm{red},k}$ by leaving the frozen
        coordinates unchanged and replacing the active coordinates
        by $T_{\mathrm{act},k}$

    \Else

        \State $T_{\mathrm{red},k}\gets I$

    \EndIf

    \State Update the reconstruction operator:
    \Statex \hspace{\algorithmicindent}
    $T_k
    \gets
    T_{\mathrm{aug},k}T_{\mathrm{red},k}$

\EndFor

\State Set
$T_{\mathrm{hybrid}}\gets T_m$

\State Verify
\Statex \hspace{\algorithmicindent}
$\|C_{\mathrm{tot}}T_{\mathrm{hybrid}}\|_F
\leq\varepsilon$

\Statex \hspace{\algorithmicindent}
$T_{\mathrm{hybrid}}\geq-\varepsilon$

\Statex \hspace{\algorithmicindent}
$\|T_{\mathrm{hybrid}}\mathbf{1}-\mathbf{1}\|_2
\leq\varepsilon$

\State \Return $T_{\mathrm{hybrid}}$

\end{algorithmic}

By construction, the resulting operator satisfies the prescribed
interface admissibility conditions and is assembled from 
non-negative reconstruction generators. Partition of unity 
is recovered through the non-negative scaling step. 
Exact geometry preservation follows from the
null-space and range relations established in Appendix~\ref{app:A}.

\section{Validation and Numerical Examples}
\label{sec:validation} 

Before applying the proposed reconstruction framework
to complex multi-patch configurations, 
we first validate its algebraic core on two simple,
well understood problems. 
Example~1 demonstrates the null‑space concept for piecewise‑linear 
elements, while Example~2 shows that, for the clamped two‑element 
quadratic case, the proposed construction reproduces exactly the 
classical B\'ezier extraction operator. 
Subsequently, Example~3 presents a more realistic multi‑patch, 
multi‑degree NURBS curve with local h‑refinements. 
Example~4 demonstrates the capability of the proposed hybrid basis 
to solve a linear boundary-value problem, whereas Example~5 solves
a nonlinear one applying degree elevation.



\textbf{EXAMPLE~1: Piecewise-linear basis functions}

The most trivial case is to consider an assemblage of piecewise-linear elements and their desired conversion into a smaller number of global functions. For example, let us assume $n_{ele}=5$ linear elements, as shown in Fig.~\ref{fig1}. If each element is initially considered as a separate entity, there are two coefficients per element ($U(x)=ax+b$) and thus the total number of unknowns in the five elements is ten  (i.e., $m=5,n_{loc} = 10$). In other words, there are 10 local shape functions (two per element) in the whole domain. 

Although the reduction from local 10 to global 6 DOFs is obvious (by intuition), it is instructive to follow the mathematical route to obtain this fact. Actually, the transformation from 10 to 6 may be produced by imposing four conditions (i.e., $r = 4$) for $C^0$-continuity at the four internal points (interfaces) at the locations $\xi=\frac{1}{5}, \frac{2}{5}, \frac{3}{5}, \frac{4}{5}$. Then, by subtracting the four continuity conditions from the total ten shape functions, the final number of global basis functions becomes six ($n_{hyb} = n_{loc} - r = 10 - 4=6$). 

More precisely, the imposition of $C^0$ continuity at internal points 2 to 5 leads to the following equations system (subscripts $L$ and $R$ stand for left and right, respectively):
\begin{subequations}
\renewcommand{\theequation}{\theparentequation.\alph{equation}}
\begin{align}
    U_{2L} - U_{2R}& = 0 \label{eq:U2} \\
    U_{3L} - U_{3R}& = 0 \label{eq:U3} \\
    U_{4L} - U_{4R}& = 0 \label{eq:U4} \\
    U_{5L} - U_{5R}& = 0 \label{eq:U5}
\end{align}
\end{subequations}

Obviously, the  equations system Eqs.~\eqref{eq:U2}-~\eqref{eq:U5} can be written in the following matrix form ($\mathbf{C}_{tot} \mathbf{U} = \mathbf{0}$):
\begin{equation}
    \begin{bmatrix}
        1 & -1 & 0 &  0 & 0 &  0 & 0 & 0 \\
        0 &  0 & 1 & -1 & 0 &  0 & 0 & 0 \\
        0 &  0 & 0 &  0 & 1 & -1 & 0 & 0 \\
        0 &  0 & 0 &  0 & 0 &  0 & 1 & -1
    \end{bmatrix}
    \begin{bmatrix}
        U_{2L} \\
        U_{2R} \\
        U_{3L} \\
        U_{3R} \\
        U_{4L} \\
        U_{4R} \\
        U_{5L} \\
        U_{5R}
    \end{bmatrix}
    =
    \begin{bmatrix}
        0 \\
        0 \\
        0 \\
        0
    \end{bmatrix} \,.
\end{equation}
The next stage is to choose the active DOFs for the interior of the domain, among the total eight $(U_{2L},U_{2R},U_{3L},U_{3R},U_{4L},U_{4R},U_{5L},U_{5R})$. While in the current case this task is obvious (for example, those DOFs having the subscript $L$ on the left side $(U_{2L},U_{3L},U_{4L},U_{5L})$), it is instructive to show that this is done using the $\mathrm{MATLAB}$ command \texttt{Z=null(C\_{tot},'r')}. In the current case, this command provides the matrix $\mathbf{T}_{hybrid}=Z$ (of size $8 \times 4$):
\begin{equation}
  \mathbf{T}_{hybrid}
    =
    \begin{bmatrix}
        1 &  0 & 0 & 0 \\
        1 &  0 & 0 & 0 \\
        0 &  1 & 0 & 0 \\
        0 &  1 & 0 & 0 \\
        0 &  0 & 1 & 0 \\
        0 &  0 & 1 & 0 \\
        0 &  0 & 0 & 1 \\
        0 &  0 & 0 & 1 
    \end{bmatrix} \label{eq6} \,.
\end{equation}
which satisfies the condition:
\begin{equation}
 \mathbf{C}_{tot} \mathbf{T}_{hybrid} 
 =
    \begin{bmatrix}
     0  &  0  &  0  &  0 \\
     0  &  0  &  0  &  0 \\
     0  &  0  &  0  &  0 \\
     0  &  0  &  0  &  0
    \end{bmatrix}     
\end{equation}
Note that, in this example, the rows of matrix \texttt{Z=null(C\_{tot},'r')}
in Eq.~\eqref{eq6} follow the Partition of Unity property, i.e. each of them 
sums to 1.

Then, the final basis functions associated with the four internal nodes (column-vector $\mathbf{N}_{global}=[N_2,\ldots,N_5]^{T}$ of size $4 \times 1$) are given in terms of the 
global vector (column-vector $\mathbf{N}_{local}$ of size $8 \times 1$ including 
piecewise-linear shape functions) by:
\begin{equation}
    \mathbf{N}_{global} =  \mathbf{T}_{hybrid}^{T} \mathbf{N}_{local}
    \label{eq8}
\end{equation}
By further considering the two extra DOFs associated with the ends ($\xi=0,1$), that is by encountering the two extreme local shape functions $N_1$ and $N_6$, the overall produced six global basis functions are illustrated in Fig.~\ref{fig1}.
\begin{figure}
    \centering
    
    \includegraphics[width=0.7\linewidth]{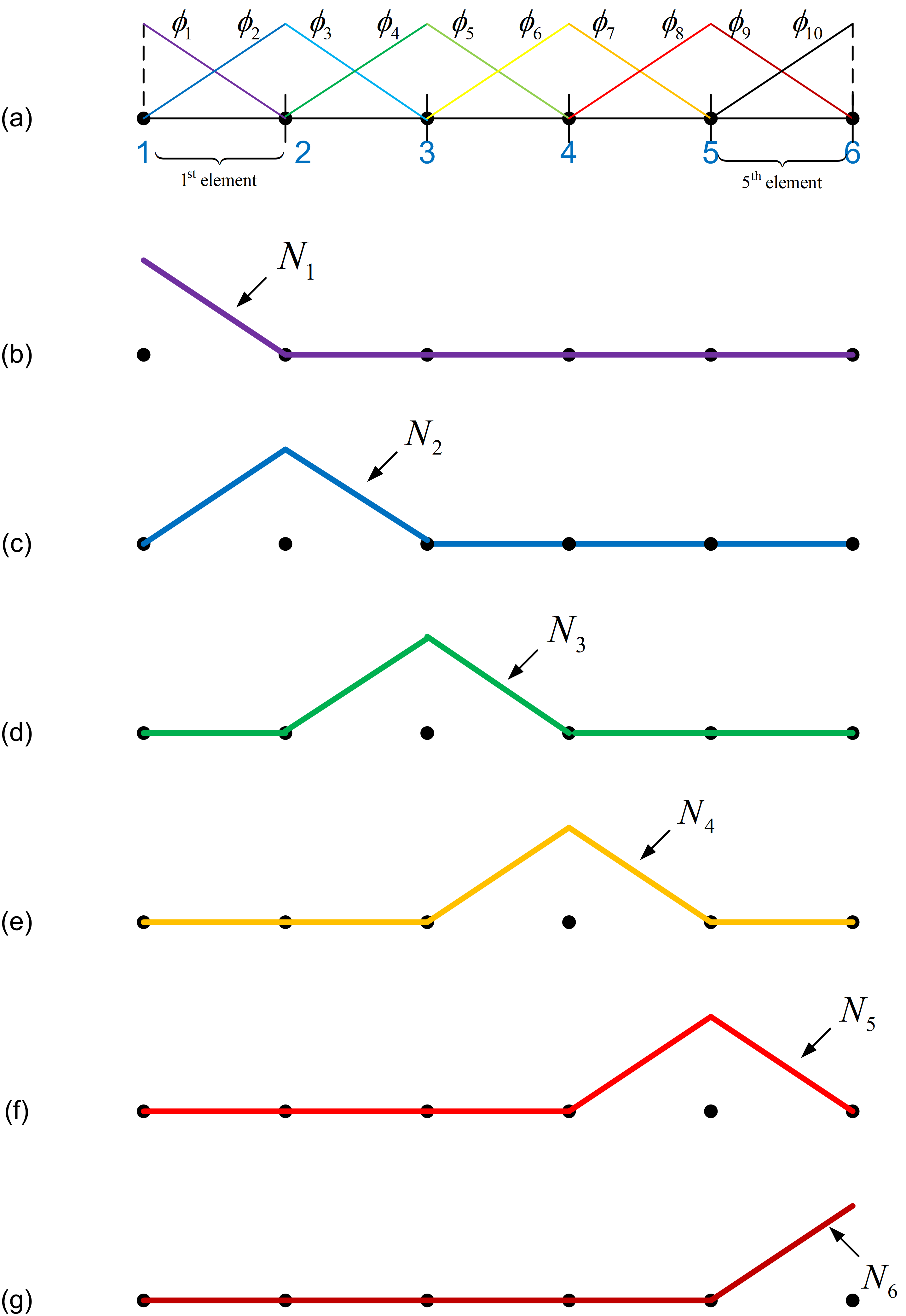} 
       \caption{(a) Ten piecewise-linear shape functions, and (b-g) the six associated global basis functions, using the $C^{0}$ null space.} 
    \label{fig1}
\end{figure}

In principle, the same philosophy can be applied to more difficult cases, where higher continuity ($C^1,C^2,\ldots$) is imposed. The main difference with the above example is that the selection of the primary DOFs is not that clear, because the spline coefficients are no more nodal values. Moreover, the $\mathrm{MATLAB}$ command \texttt{null(C\_{tot},'r')} is not always capable of producing basis functions which fulfill the partition of unity property and does not generally ensure positivity. 


\textbf{EXAMPLE-2: Null-Space Construction for Two Quadratic B\'ezier Elements}

\textbf{Problem formulation:} Without loss of generality, for the sake of brevity we consider only two quadratic B\'ezier elements sharing
the breakpoint $x=\frac12$. Initially, the two elements are completely
\textit{disconnected} and each possesses its own Bernstein basis,
\begin{equation}
\mathbf B^{(1)}
=
\begin{bmatrix}
B^{(1)}_0\\
B^{(1)}_1\\
B^{(1)}_2
\end{bmatrix},
\qquad
\mathbf B^{(2)}
=
\begin{bmatrix}
B^{(2)}_0\\
B^{(2)}_1\\
B^{(2)}_2
\end{bmatrix} \label{eq9} \,.
\end{equation}

Consequently, the disconnected approximation space consists of six basis
functions,
\begin{equation}
    \mathbf N_{\mathrm{disc}}
    =
    \begin{bmatrix}
    B^{(1)}_0\\
    B^{(1)}_1\\
    B^{(1)}_2\\
    B^{(2)}_0\\
    B^{(2)}_1\\
    B^{(2)}_2
    \end{bmatrix} \label{eq10} \,,
\end{equation}
associated with six independent degrees of freedom, i.e. $(a_0,a_1,a_2)$ and $(b_0,b_1,b_2)$ for the left and the right element, respectively.

Our objective is to recover the classical quadratic spline associated with the
knot vector
\begin{equation}
    \Xi=
    \left[
    0,0,0,\frac12,1,1,1
    \right] \label{eq11} \,,
\end{equation}
which possesses only four global basis functions after imposing $C^0$ and
$C^1$ continuity at the interface.

\bigskip

\noindent\textbf{Continuity constraints:} Since the two B\'ezier elements are initially disconnected, continuity must be enforced explicitly. The $C^0$ continuity condition requires the function values to coincide at the common endpoint (at $x=\frac{1}{2}$ where $B_2^{(1)}=B_0^{(2)}=1$ and the rest four Bernstein polynomials vanish),
\begin{equation}
    a_2 = b_0  \label{eq12} \,,
\end{equation}
whereas the $C^1$ continuity condition requires equality of the first
derivatives,
\begin{equation}
    -a_1+a_2+b_0-b_1=0 \label{eq13} \,.
\end{equation}

Collecting the above equations yields the continuity system

\begin{equation}
    \begin{bmatrix}
    0&0&1&-1&0&0\\
    0&-1&1&1&-1&0
    \end{bmatrix}
        \begin{bmatrix}
        a_0\\
        a_1\\
        a_2\\
        b_0\\
        b_1\\
        b_2
        \end{bmatrix}
    =
    \mathbf0 \label{eq14} \,.
\end{equation}

Denoting the above constraint matrix by $C$, the admissible coefficient vectors must satisfy
\begin{equation}
    C\mathbf a=\mathbf0 \label{eq15} \,.
\end{equation}
with
\begin{equation}
    \mathbf{a}
    =
       \begin{bmatrix}
        a_0\\
        a_1\\
        a_2\\
        b_0\\
        b_1\\
        b_2
        \end{bmatrix}  
    \label{eq16} \,.
\end{equation}

\bigskip

\noindent\textbf{Null-space basis:} The admissible solution space is obtained by computing the null space of the constraint matrix, which in MATLAB this is accomplished by the command \texttt{Z = null(C,'r')}.


The function \texttt{null} computes a basis of the kernel
\begin{equation}
    \ker(C)
    =
    \{\mathbf x\;|\;C\mathbf x=\mathbf0\} \label{eq18} \,,
\end{equation}
returning four linearly independent vectors spanning the admissible space.
Using the option \texttt{'r'} produces a rational basis whenever possible,
thereby avoiding floating-point round-off errors.

For the present example, the command \texttt{Z = null(C,'r')} gives
\begin{equation}
    Z=
    \begin{bmatrix}
    1&0&0&0\\
    0&2&-1&0\\
    0&1&0&0\\
    0&1&0&0\\
    0&0&1&0\\
    0&0&0&1
    \end{bmatrix} \label{eq19} \,,
\end{equation}
which, by definition, satisfies
\begin{equation}
    CZ=\mathbf0 \label{eq20} \,.
\end{equation}

Note that, similarly to Example~1, again in Example~2 the rows of matrix \texttt{Z=null(C\_{tot},'r')} in Eq.~\eqref{eq19} algebraically sum to 1. Nevertheless, the function \texttt{Z=null(C\_{tot},'r')} does \textit{not} always ensure the Partition of Unity property.

Each column of $Z$, given by Eq.~\eqref{eq19}, represents one admissible coefficient vector satisfying the prescribed continuity constraints. Consequently, the four coupled basis functions are obtained by applying Eq.~(8), in which the role of the transformation matrix $\mathbf{T}_{hybrid}$ is temporarily played by $Z$ itself.

Figure~\ref{fig2} illustrates the resulting basis functions, $\Phi_1$ to $\Phi_4$, which were found to satisfy the Partition of Unity (PoU) property. Nevertheless, they take negative values as well (obviously, because the matrix $Z$ in Eq.~\eqref{eq19} includes the negative entry $-1$ in its third column). This is exactly the disadvantage of the MATLAB function \texttt{null}, for which a remedy becomes imperative. The updated matrix $\mathbf{T}_{hybrid}$ will include linear combinations of rows and columns in $Z$, which not only will lead to non-negative basis functions, but also they must fulfill the PoU property. Next, a simplified procedure for the change of matrix $Z$ to the desired form $\mathbf{E}=\mathbf{T}_{hybrid}$ is demonstrated.
\begin{figure}[!ht]
\centering
\includegraphics[width=0.72\linewidth]{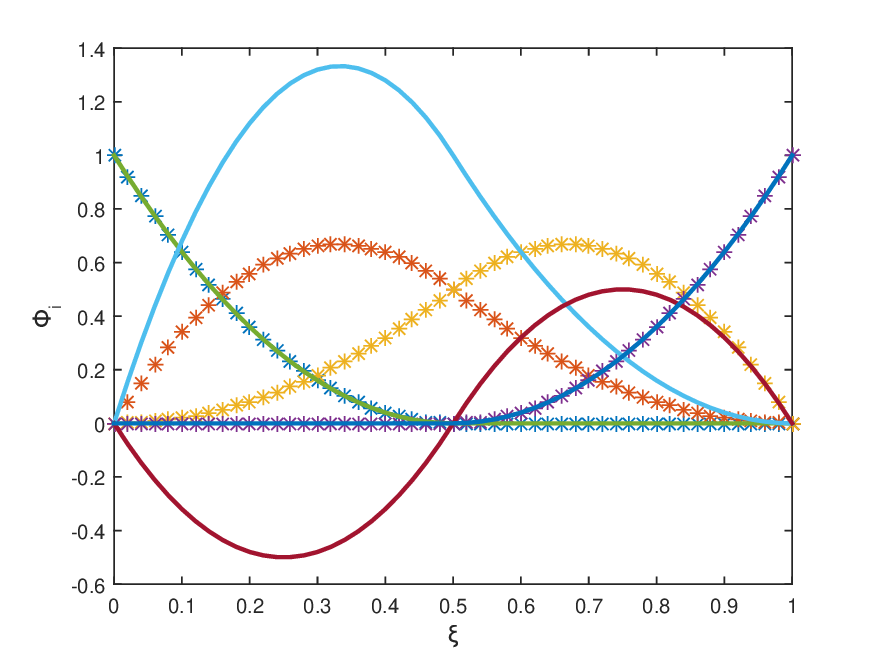} 
\caption{The four coupled basis functions obtained directly from the null-space matrix $Z=\mathrm{null}(C,\texttt{'r'})$ (solid lines), compared with the exact Cox-de Boor B-splines (starred lines).}
\label{fig2}
\end{figure}

In short-hand, we can write:
\begin{equation}
    \mathbf{Z} = [ \mathbf{z}_1, \mathbf{z}_2, \mathbf{z}_3,\mathbf{z}_4 ] \label{eq21} \,,
\end{equation}
where $\mathbf{z}_i,i=1,\ldots,4$ is the $i$th column of matrix $\mathbf{Z}$ (see, Eq.~\eqref{eq19}).

The entire procedure is performed in three steps as follows. In the first step we identify the column of matrix $\mathbf{Z}$ which include negative entries. In the second step we update these columns so that only positive entries appear, whereas in the third step we determine scaling factors that ensure the Partition of Unity property in all rows of the updated $\mathbf{Z}$ matrix, called $\mathbf{E}=\mathbf{T}_{hybrid}$.

Step-1: Equation~\eqref{eq19} shows that only the third column of matrix $\mathbf{Z}$ incudes negative entries, and thus
 \begin{equation}
     \mathbf{r}_1=\mathbf{z}_1, \quad \mathbf{r}_2=\mathbf{z}_2, \quad \mathbf{r}_4=\mathbf{z}_4 \label{eq22} \,.
 \end{equation}

 Step-2: To obtain $\mathbf{r}_3$, the third column of matrix $\mathbf{Z}$ is replaced by a linear combination of 2nd and 3rd column, which can be written as follows:
\begin{equation}
     \mathbf{x} = a \mathbf{z}_2 + b \mathbf{z}_3
       = a \begin{bmatrix}
           0 \\
           2 \\
           1 \\ 
           1 \\
           0 \\
           0
       \end{bmatrix}  
       + b \begin{bmatrix}
           0 \\
          -1 \\
           0 \\ 
           0 \\
           1 \\
           0
       \end{bmatrix}  
       = \begin{bmatrix}
           0 \\
           2a-b \\
           a \\ 
           a \\
           b \\
           0
       \end{bmatrix}  
     \label{eq23} \,.
\end{equation}

Considering the extreme ray, we have either $b=0$ or $2a-b=0$. Since the first condition repeats the already known $r_2$, we stick on the second condition, i.e., $2a-b=0$, whence $b=2a$. Setting, for example, $a=1$ we obtain the simple expression 

\begin{equation}
    \mathbf{r}_3=\begin{bmatrix}
           0 \\
           0 \\
           1 \\ 
           1 \\
           2 \\
           0
       \end{bmatrix}  \label{eq24} \,.
\end{equation}

Therefore, Eq.~\eqref{eq22} and Eq.~\eqref{eq24} suggest the following nonnegative vectors
\begin{equation}
    \mathbf{r}_1=\begin{bmatrix}
           1 \\
           0 \\
           0 \\ 
           0 \\
           0 \\
           0
       \end{bmatrix},   
    \mathbf{r}_2=\begin{bmatrix}
           0 \\
           2 \\
           1 \\ 
           1 \\
           0 \\
           0
       \end{bmatrix},  
    \mathbf{r}_3=\begin{bmatrix}
           0 \\
           0 \\
           1 \\ 
           1 \\
           2 \\
           0
       \end{bmatrix},  
       \mathbf{r}_4=\begin{bmatrix}
           0 \\
           0 \\
           0 \\ 
           0 \\
           0 \\
           1
       \end{bmatrix}   
       \label{eq25} \,.
\end{equation}

Step-3: Now we set the transformation matrix in the form
\begin{equation}
    \mathbf{E}=[\mathbf{e}_1,\mathbf{e}_2,\mathbf{e}_3,\mathbf{e}_4] \label{eq26} \,,
\end{equation}
with the columns $\mathbf{e}_i$'s being multiples of the above $\mathbf{r}_i$'s:
\begin{equation}
    \mathbf{e}_1=\alpha_1 \mathbf{r}_1, \quad  
    \mathbf{e}_2=\alpha_2 \mathbf{r}_2, \quad  
    \mathbf{e}_3=\alpha_3 \mathbf{r}_3, \quad  
    \mathbf{e}_4=\alpha_4 \mathbf{r}_4   
       \label{eq27} \,,
\end{equation}
where 
the scaling factors $(\alpha_1,\alpha_2,\alpha_3,\alpha_4)$ are to be determined.

The determination of the above-mentioned scaling factors is made by imposing the Partition of Unity condition to each of the six rows of matrix $\mathbf{E}$, i.e.:
\begin{equation}
    \sum_{j=1}^{4} (\mathbf{e}_j)_i = 1, \qquad i = 1, \ldots, 6
    \label{eq29} \,.
\end{equation}

The implementation of Eq.~\eqref{eq29} leads to the following overdetermined system:
\begin{equation}
    \begin{bmatrix}
       1 & 0 & 0 & 0 \\
       0 & 2 & 0 & 0 \\
       0 & 1 & 1 & 0 \\
       0 & 1 & 1 & 0 \\
       0 & 0 & 2 & 0 \\ 
       0 & 0 & 0 & 1
    \end{bmatrix}
    \begin{bmatrix}
       \alpha_1  \\
       \alpha_2 \\
       \alpha_3 \\
       \alpha_4
    \end{bmatrix}
    =
    \begin{bmatrix}
        1 \\
        1 \\
        1 \\
        1 \\
        1 \\
        1
    \end{bmatrix}
    \label{eq30} \,.
\end{equation}
The unique solution of Eq.~\eqref{eq30} is
\begin{equation}
    \alpha_1=1, \quad \alpha_2=\alpha_3=\frac{1}{2}, \quad \alpha_4=1 \label{eq31} \,.
\end{equation}
Obviously, in the general case a least-squares method is required, for example, using the QR-factorization algorithm.

Substituting Eq.~\eqref{eq31} into Eq.~\eqref{eq27}, we eventually obtain the transformation matrix:
\begin{equation}
\mathbf{E}
    =
    \begin{bmatrix}
        1 & 0 & 0 & 0 \\
        0 & 1 & 0 & 0 \\
        0 & \frac{1}{2} & \frac{1}{2} & 0 \\
        0 & \frac{1}{2} & \frac{1}{2} & 0 \\
        0 & 0 & 1 & 0 \\
        0 & 0 & 0 & 1
    \end{bmatrix}
    \label{eq32} \,.
\end{equation}
One may verify the following two facts:
\begin{itemize}
    \item The matrix $\mathbf{E}$ is a linear combination of the initial null space matrix $\mathbf{Z}$, i.e., $\mathbf{E}=\mathbf{Z}\mathbf{M}$, with transformation matrix:
    \begin{equation}
\mathbf{M} = 
\begin{bmatrix}
1 & 0 & 0 & 0\\
0 & 1/2 & 1/2 & 0\\
0 & 0 & 1 & 0\\
0 & 0 & 0 & 1
\end{bmatrix}.
\end{equation}
    \item In current case, the matrix $\mathbf{E}$ is identical with the B\'ezier extraction operator $\mathbf{C}_e$ associated with the knot vector of Eq.~\eqref{eq11}.
\end{itemize}
Therefore, the four B-spline functions can be accurately calculated applying  Eq.~\eqref{eq8}, in which the transformation matrix is given as $\mathbf{T}_{hybrid}=\mathbf{E}$. Then, the produced numerical values are nonnegative and coincide with those of the Cox-de Boor, i.e., the strarred lines in Fig.~\ref{fig2}.

This completes the manipulation of Example~2, from which we learn that the algebraic basis $Z=\operatorname{null}(C,\texttt{'r'})$ cannot be used directly as the transformation matrix. It is merely a starting point; it may contain negative values and does not generally enforce the partition of unity. The correct basis $E$ is obtained by extracting the extreme rays of the non-negative cone with prescribed supports, followed by a global scaling via the element-wise partition of unity. In this example, the procedure reproduces the standard operator exactly, illustrating that the desired physical properties (non-negativity and partition of unity) must be imposed algebraically rather than expected from the raw null-space computation.




\textbf{Interim Remarks on Example~2}

Although the proposed method utilizes unclamped local B-splines per element (and not Bernstein polynomials), the results of Example~2 were very instructive, and thus some additional remarks are provided below.

\textbf{Remark-1:} The determination of the scaling factors $\alpha$ leads to the linear system
\begin{equation}
A_\alpha \,\alpha = \mathbf{1}_{m(p+1)},
\end{equation}
which is generally overdetermined (cf. Eq.~\eqref{eq30}), as the number of rows $m(p+1)$ (one equation for each Bernstein coefficient in each element) typically exceeds the number of unknowns $r = m+p$. Nevertheless, the system is theoretically consistent, meaning that an exact solution exists due to the inherent structure of the B-spline basis and the consistency of the partition-of-unity constraints. In practice, the system is solved using the least-squares method, which minimises the residual $\|A_\alpha \alpha - \mathbf{1}\|_2$. This is conveniently performed in numerical computing environments using the backslash operator (e.g., \texttt{$\alpha = A_\alpha \setminus \mathbf{1}$} in MATLAB), which employs a stable QR decomposition with column pivoting. The least-squares solution is unique, provided $A_\alpha$ has full column rank, and it yields the correct scaling factors that enforce the element-wise partition of unity exactly (up to machine precision).

\bigskip
\textbf{Remark-2:} While the above procedure is generally applicable, it is under question whether MATLAB \texttt{null(C,'r')} function, i.e. the computation of the \textit{entire} null-space, is efficient in large-scale problems. For large‑scale applications, the algebraic extraction operator $E$ is computed by solving the sparse linear system (13) directly, rather than by constructing $Z$ and extracting rays. This approach avoids the computational cost of null‑space computation and ray enumeration, and it is readily parallelisable. The two‑element example presented here serves as a validation of the underlying algebraic principles, which are then implemented in the general sparse solver for practical use.

The overdetermined linear system (14) is solved in the least‑squares sense to determine the scaling factors $\alpha$. To ensure numerical robustness, the solution is computed via QR factorization with column pivoting, which avoids the explicit formation of $A_{\alpha}^{T} A_{\alpha}$ and preserves the conditioning of the problem. In MATLAB, this is conveniently performed by the backslash operator, which internally employs a QR algorithm for overdetermined systems.

In practice, for large‑scale problems, we do not compute the null‑space basis $Z$. Instead, we assemble the sparse linear system (Eq. X) directly from the continuity, support, and partition‑of‑unity constraints, and solve it for vec(E) using a sparse QR solver. This approach is numerically stable and avoids the computational cost of null‑space computation and ray enumeration."

\textbf{Column-wise algebraic construction of the extraction operator:} To avoid the computational cost and numerical complexity of computing the full null-space basis \(Z=\operatorname{null}(C,\texttt{'r'})\), we adopt a column-wise algebraic strategy that constructs the Bézier extraction operator \(E\) directly from the continuity constraints and the known supports of the B-spline basis functions.

For each B-spline basis function \(j=1,\dots,r\), let \(\mathcal{S}_j \subset \{1,\dots,n\}\) be the set of Bernstein coefficient indices that lie within its element support (known from the knot vector). We first build the restricted homogeneous system
\begin{equation}
C x = 0, \qquad x_i = 0 \quad \forall \, i \notin \mathcal{S}_j.
\tag{1}
\end{equation}
Let \(N_j \in \mathbb{R}^{n \times d_j}\) be a basis of the solution space of (1). The unnormalised column \(r_j\) is the unique extreme ray of the cone
\begin{equation}
\mathcal{K}_j = \left\{ x = N_j y \;\middle|\; y \in \mathbb{R}^{d_j},\; x \ge 0 \right\}.
\tag{2}
\end{equation}
This extreme ray corresponds to the minimal-support, non-negative solution and is found by enumerating the basic feasible solutions of (2). In practice, for small problems we set \(d_j-1\) variables to zero; for large-scale applications, this step is replaced by a call to a linear programming solver (e.g., \texttt{linprog} in MATLAB). Collecting these rays gives the matrix of unnormalised columns
\begin{equation}
R = [\,r_1,\; r_2,\; \dots,\; r_r\,] \in \mathbb{R}^{n \times r}.
\tag{3}
\end{equation}

The columns of \(R\) have the correct directions but arbitrary scales. To determine the unique scaling, we enforce the element-wise partition of unity. Let \(P_k \in \mathbb{R}^{(p+1)\times n}\) select the Bernstein coefficients of the \(k\)-th element. We seek scaling factors \(\alpha = [\alpha_1,\dots,\alpha_r]^T\) such that
\begin{equation}
P_k \, R \, \operatorname{diag}(\alpha) \, \mathbf{1}_r = \mathbf{1}_{p+1}, \qquad k=1,\dots,m,
\tag{4}
\end{equation}
where \(\mathbf{1}_r\) and \(\mathbf{1}_{p+1}\) are vectors of ones of appropriate dimensions. Equation (4) is a linear system for \(\alpha\):
\begin{equation}
A_\alpha \, \alpha = \mathbf{1}_{m(p+1)},
\tag{5}
\end{equation}
which is overdetermined but consistent. We solve it in the least-squares sense using QR factorisation with column pivoting (e.g., the backslash operator in MATLAB), which ensures numerical stability by avoiding the explicit formation of \(A_\alpha^T A_\alpha\).

Finally, the Bézier extraction operator is assembled as
\begin{equation}
E = R \, \operatorname{diag}(\alpha).
\tag{6}
\end{equation}
The resulting matrix \(E\) satisfies \(C E = 0\), is component-wise non-negative, and fulfills the element-wise partition of unity \(P_k E \mathbf{1}_r = \mathbf{1}_{p+1}\) for every element \(k\). By the uniqueness of the minimal-support non-negative basis, this \(E\) coincides with the standard Bézier extraction operator.

This column-wise procedure completely avoids the computation of the full null-space basis \(Z\), and its cost scales linearly with the number of basis functions \(r\). It is therefore suitable for both small illustrative examples and large-scale practical implementations.


\bigskip
\textbf{Remark-3:} Instead of calculating the entire null-space of matrix C (by applying MATLAB \texttt{null} command on the entire matrix C), we focus on each separate column of the temporal matrix $\mathbf{R}=[r_1,\ldots,r_4]$ imposing the local support and eventually computing a local (column-wise) null space. Next, after the matrix $R$ has been computed, we follow the same procedure as previously.

Next we produce the temporal matrix $\mathbf{R}$.

\textbf{Step 1:} We create the first column (left end)
We know that 1st B‑spline lives only inside 1st element. Therefore, the positions 4,5,6 (of 2nd element) must be zero. Solving the system $Cx=0$ with $x_4=x_5=x_6=0$, we receive the unique solution 
\begin{equation}
    r_1=[1, 0, 0, 0, 0, 0] \,.
\end{equation}

\textbf{Step 2:} We create the 4th column of $R$ (right end)
We know that the 4th B‑spline lives only inside the 2nd element. Therefore, the positions 1,2,3 (of 1st element) will be zero. Solving the system $Cx=0$ with $x_1=x_2=x_3=0$, we receive the unique solution 
\begin{equation}
    r_4=[0, 0, 0, 0, 0, 1] \,.
\end{equation}

\textbf{Step3:} We create the 2nd column (internal, left). 
We know that the 2nd B‑spline lives within both elements, but it does not touch neither the left end of the first element (position 1) nor the right end of the 2nd element (positions 5 and 6). Solving the system $Cx=0$ with $x_1=x_5=x_6=0$, we receive the unique solution 
\begin{equation}
    r_2=[0, 2, 1, 1, 0, 0] \label{eq37} \,.
\end{equation}
Interestingly, the same result may be obtained by constructing the restricted constraint matrix" ($A_{eq}^{(2)}$ of size $5\times6$), of which the first two rows correspond to the imposed continuity conditions whereas the rest three include the locality restrictions $x_1=x_5=x_6=0$, and thus:
\begin{equation}
    A_{eq}^{(2)}=
    \begin{bmatrix}
        0 &0 &1 &-1 &0 &0 \\
        0 &-1 &1 &1 &-1 &0 \\
        1 &0 &0 &0 &0 &0 \\
        0 &0 &0 &0 &1 &0 \\
        0 &0 &0 &0 &0 &1
    \end{bmatrix}
\end{equation}
In this case, the output of MATLAB command \texttt{null(Aeq2,'r')} is exactly that given by Eq.~\eqref{eq37}.

\paragraph{Step 4: Construction of the third column (interior, right).}
The third B-spline basis function has support over both elements, but it does not touch the left edge of the first element (positions 1 and 2) nor the right edge of the second element (position 6). We therefore impose
\[
x_1 = 0, \qquad x_2 = 0, \qquad x_6 = 0,
\]
and solve the homogeneous system \(C x = 0\). From the first row of \(C\), we obtain
\[
x_3 - x_4 = 0 \quad \Longrightarrow \quad x_3 = x_4 = t.
\]
Substituting into the second row of \(C\) (with \(x_2 = 0\)) gives
\[
4x_3 + 4x_4 - 4x_5 = 0 \quad \Longrightarrow \quad x_5 = x_3 + x_4 = 2t.
\]
Thus, for \(t=1\), the unique (up to scaling) direction is
\[
r_3 = [\,0,\; 0,\; 1,\; 1,\; 2,\; 0\,]^T,
\]
which is non-negative. No additional transformation is required to obtain this column.

\bigskip




Again, it should be made clear that the transformation from Bernstein polynomials to B-splines (of constant degree) is not of practical value for the proposed methodology, since this task can be performed easily and more efficiently through the B\'{e}zier extraction operator. Nevertheless, having validated the proposed methodology in Example~2, we are now ready to proceed with more realistic and practical examples, which are presented below.


\bigskip

\noindent\textbf{EXAMPLE-3: A multi‑patch, multi‑degree NURBS curve with local h‑refinements} 

This is an example in which the proposed method is thoroughly implemented. First the set of hybrid basis functions is constructed and its properties are verified. Second the ability of this functional set to handle weights and accurately representing (with machine accuracy) a circular arc is shown (see Result 5: Geometry reconstruction). 

The initial NURBS curve is of degree $p=2$ with the clamped knot vector
\[
\Xi_{\text{init}} = \{0,0,0,\;0.2,\;0.4,\;0.6,\;0.8,\;1,1,1\},
\]
which defines $7$ global basis functions over the interval $[0,1]$.

The curve is decomposed into $5$ unclamped patches, each corresponding to one of the initial knot spans:
\[
[0,0.2],\quad [0.2,0.4],\quad [0.4,0.6],\quad [0.6,0.8],\quad [0.8,1].
\]

Subsequently, we perform the following local operations:

\begin{enumerate}
\item \textbf{Local degree elevation:} Patches 2 and 3 are elevated from degree $2$ to degree $3$. Patches 1, 4, and 5 retain degree $2$.

\item \textbf{Local h-refinement (knot insertion):}
\begin{itemize}
\item \textbf{Patch 1} ($p=2$, interval $[0,0.2]$): one knot is inserted at its middle $\xi=0.1$ with multiplicity $1$. The local knot vector, initially consisting of $2p+2$ entries, becomes
\[
\Xi_1 = \{0,0,0,\;0.1,\;0.2,0.2,0.2\}.
\]

\item \textbf{Patch 2} ($p=3$, interval $[0.2,0.4]$): one knot is inserted at its middle $\xi=0.3$ with multiplicity $1$. The local knot vector becomes
\[
\Xi_2 = \{0.2,0.2,0.2,0.2,\;0.3,\;0.4,0.4,0.4,0.4\}.
\]

\item \textbf{Patch 3} ($p=3$, interval $[0.4,0.6]$): the patch is refined twice, which inserts three knots in total: first at $\xi=0.5$ (multiplicity $1$), then at $\xi=0.45$ and $\xi=0.55$ (each multiplicity $1$). The local knot vector becomes
\[
\Xi_3 = \{0.4,0.4,0.4,0.4,\;0.45,\;0.5,\;0.55,\;0.6,0.6,0.6,0.6\}.
\]

\item \textbf{Patch 4} ($p=2$, interval $[0.6,0.8]$): no refinement. The local knot vector is
\[
\Xi_4 = \{0.6,0.6,0.6,\;0.8,0.8,0.8\}.
\]

\item \textbf{Patch 5} ($p=2$, interval $[0.8,1]$): no refinement. The local knot vector is
\[
\Xi_5 = \{0.8,0.8,0.8,\;1,1,1\}.
\]
\end{itemize}
\end{enumerate}

The resulting local patch dimensions (number of basis functions per patch) are
\[
n = [4,\;5,\;7,\;3,\;3],
\]
with corresponding degrees
\[
p = [2,\;3,\;3,\;2,\;2].
\]
There are $N=5$ patches, hence $4$ interfaces:
\[
\xi = 0.2,\; 0.4,\; 0.6,\; 0.8.
\]
At each interface we impose $C^1$ continuity (i.e., both $C^0$ and $C^1$), which yields $2$ constraints per interface, or $8$ constraints in total. The total number of degrees of freedom for the final hybrid basis is therefore
\[
r = \sum n_i - 2\cdot(N-1) = (4+5+7+3+3) - 8 = 22 - 8 = 14.
\]

This data set completely defines the multi‑patch, multi‑degree spline space that serves as input for the algebraic merging algorithm.

\textbf{Further Results for Example~3:}

We first examine the state of the basis before the algebraic merging is applied. Figure~\ref{fig:decomp_basis} shows the 22 local basis functions obtained after decomposition and local refinements (patch dimensions $n=[4,5,7,3,3]$). At this stage, no continuity constraints have been enforced across the four interfaces. The discontinuities at $\xi = 0.2, 0.4, 0.6,$ and $0.8$ are clearly visible: the basis functions from adjacent patches do not connect smoothly. This $C^{-1}$ space is the direct output of the local patch construction and serves as the input to the merging algorithm.
\begin{figure}[ht!]
\centering
\includegraphics[width=1.00\linewidth]{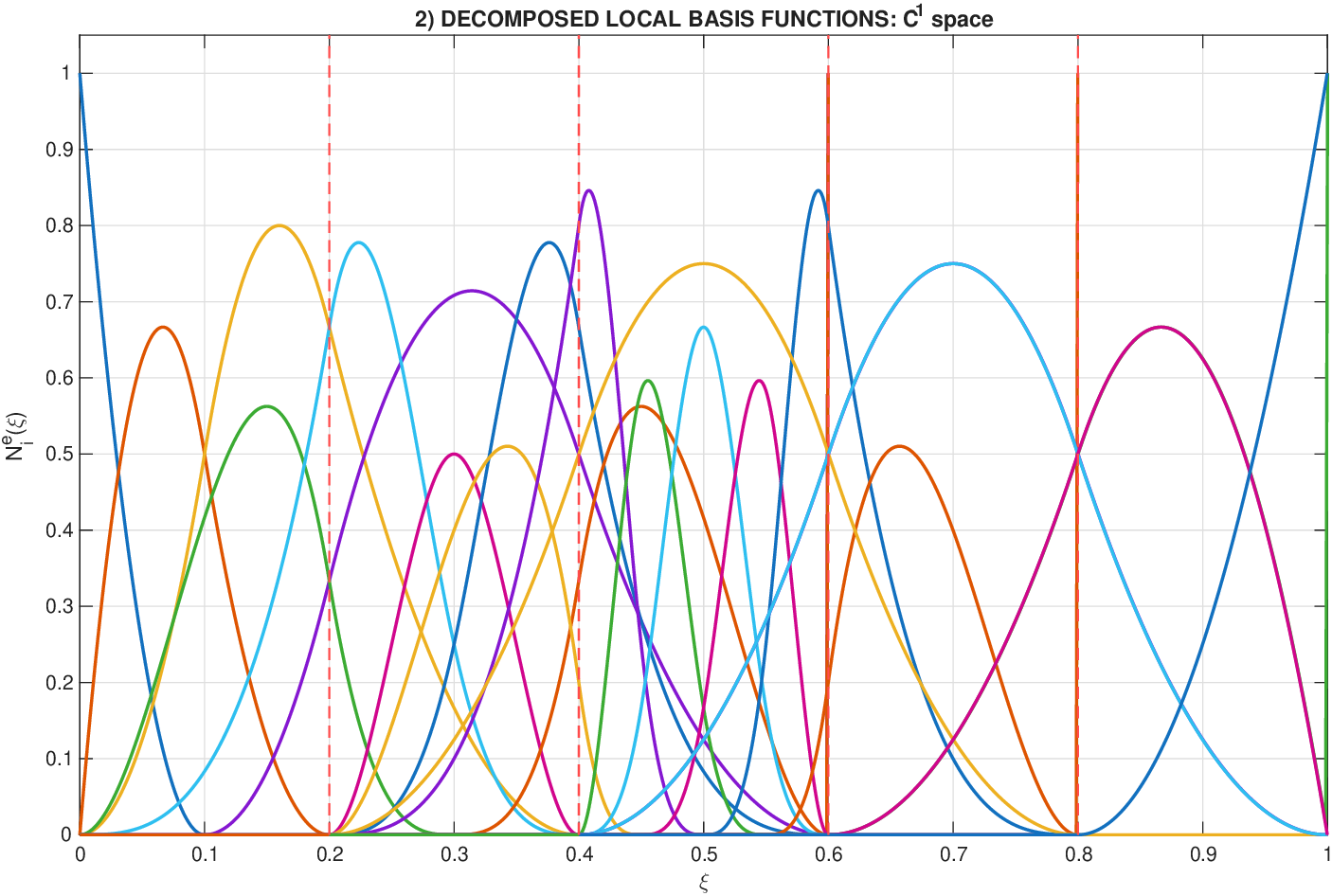} 
\caption{Decomposed local basis functions before merging ($C^{-1}$ space). The 22 local functions are discontinuous at the interfaces $\xi = 0.2, 0.4, 0.6,$ and $0.8$, as they come from independent unclamped patches with no continuity enforcement.}
\label{fig:decomp_basis}
\end{figure}

In contrast, Figure~\ref{fig:hybrid_basis} shows the final 14 hybrid basis functions after the algebraic merging. The discontinuities have been removed, and the functions now form a $C^1$-continuous, non‑negative, partition‑of‑unity basis over the entire domain. The reduction from 22 to 14 basis functions corresponds exactly to the $2$ constraints imposed per interface ($4$ interfaces $\times 2 = 8$ constraints), confirming the expected dimension of the global space.
\begin{figure}[ht!]
\centering
\includegraphics[width=1.00\linewidth]{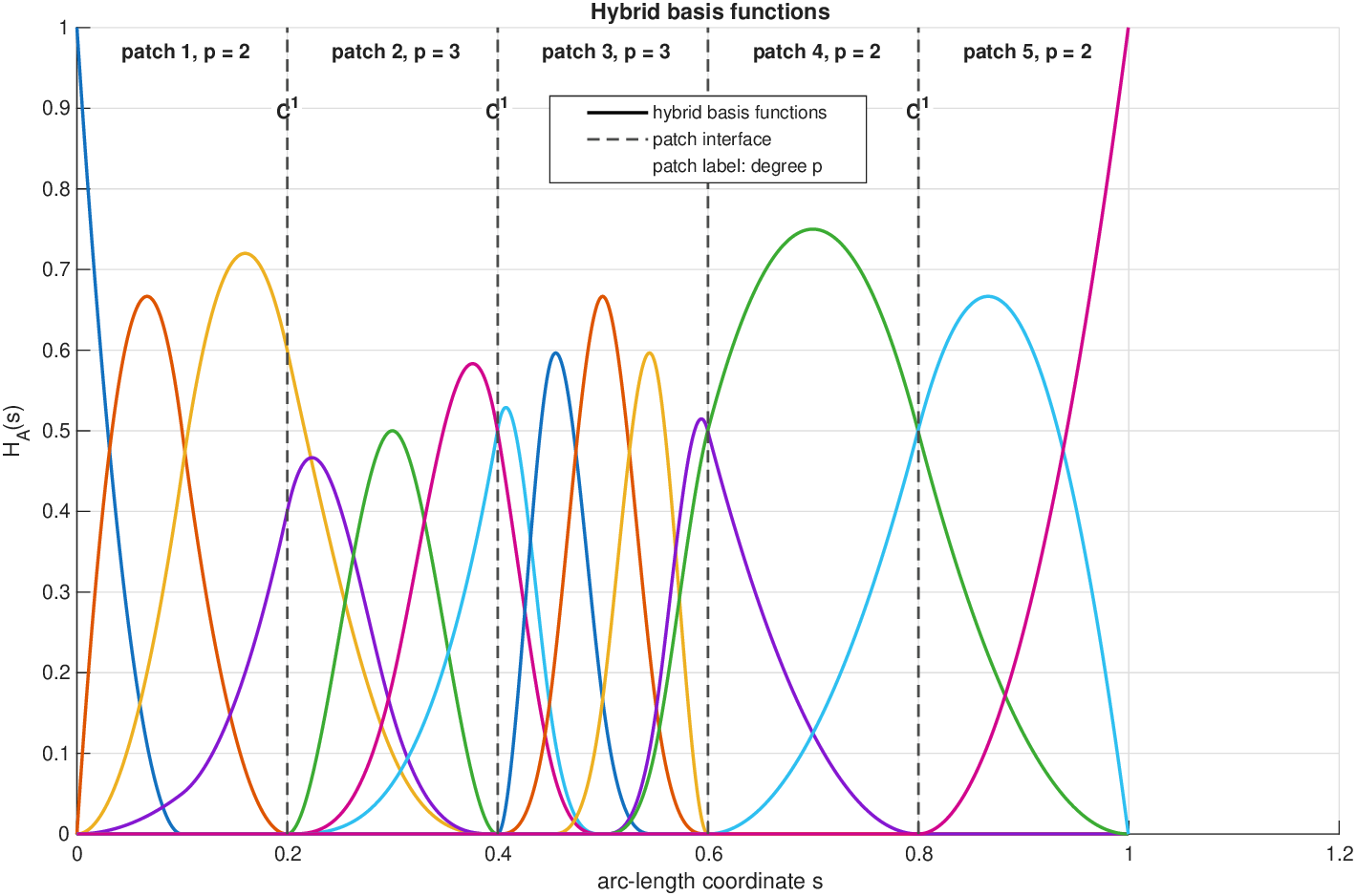} 
\caption{The fourteen basis functions.}
\label{fig:hybrid_basis}
\end{figure}

More details regarding the hybrid basis functions are provided below. 

\paragraph{Result 1: Hybrid basis functions.}
Figure~\ref{fig:hybrid_basis} shows the 14 basis functions of the final global basis using $T_{\text{hybrid}}$. All functions are non‑negative and satisfy the partition of unity. At each of the four interfaces ($\xi=0.2,0.4,0.6,0.8$), the functions from adjacent patches join with $C^1$ continuity, as confirmed by the continuity check below.

\paragraph{Local support of the hybrid basis.}
After constructing the 14 global basis functions, it is instructive to examine their local support on the refined mesh. Table~\ref{tab:active_basis} lists, for each of the 10 elements, which of the 14 basis functions are active (i.e., have at least one non‑zero entry in the corresponding rows of $T_{\text{hybrid}}$). The number of active functions varies from 2 to 5, reflecting both the varying polynomial degrees ($p=2$ for patches 1,4,5 and $p=3$ for patches 2,3) and the effect of the $C^1$ couplings introduced by the algebraic merging. This local support pattern confirms that the basis is sparse and inherently local, which is essential for computational efficiency.

\begin{table}[htbp]
\centering
\caption{Active global basis functions per element.}
\begin{tabular}{c|c|c}
Element & Active basis functions & Count \\ \hline
$E_1=[0,0.1]$   & $\phi_1,\phi_2,\phi_3,\phi_4$         & 4 \\
$E_2=[0.1,0.2]$  & $\phi_2,\phi_3,\phi_4$                & 3 \\
$E_3=[0.2,0.3]$  & $\phi_3,\phi_4,\phi_5,\phi_6,\phi_7$  & 5 \\
$E_4=[0.3,0.4]$  & $\phi_3,\phi_4,\phi_5,\phi_6,\phi_7$  & 5 \\
$E_5=[0.4,0.45]$ & $\phi_6,\phi_7,\phi_8,\phi_9$         & 4 \\
$E_6=[0.45,0.5]$ & $\phi_6,\phi_7,\phi_8,\phi_9,\phi_{10}$ & 5 \\
$E_7=[0.5,0.55]$ & $\phi_8,\phi_9,\phi_{10},\phi_{11},\phi_{12}$ & 5 \\
$E_8=[0.55,0.6]$ & $\phi_9,\phi_{10},\phi_{11},\phi_{12}$ & 4 \\
$E_9=[0.6,0.8]$  & $\phi_{11},\phi_{12},\phi_{13}$       & 3 \\
$E_{10}=[0.8,1]$ & $\phi_{13},\phi_{14}$                 & 2 \\
\end{tabular}
\label{tab:active_basis}
\end{table}

\paragraph{Result 2: Continuity verification.}
For each interface, we evaluate the jump in value and in first derivative 
for every basis function. The maximum jumps are reported in Table~\ref{tab:continuity}. 
All jumps are below machine precision, confirming that the basis fulfills 
the prescribed $C^1$ continuity, as illustrated in Fig.~\ref{fig6}.

\begin{table}[htbp]
\centering
\begin{tabular}{c|c|c}
Interface $\xi$ & $\max| \Delta u |$ & $\max| \Delta u' |$ \\
\hline
0.2 & $1.11\times 10^{-16}$ & $2.22\times 10^{-16}$ \\
0.4 & $1.11\times 10^{-16}$ & $1.11\times 10^{-16}$ \\
0.6 & $2.22\times 10^{-16}$ & $3.33\times 10^{-16}$ \\
0.8 & $1.11\times 10^{-16}$ & $1.11\times 10^{-16}$ \\
\end{tabular}
\caption{Maximum continuity jumps at interfaces.}
\label{tab:continuity}
\end{table}

\begin{figure}[ht!]
\centering
\includegraphics[width=1.1\linewidth]{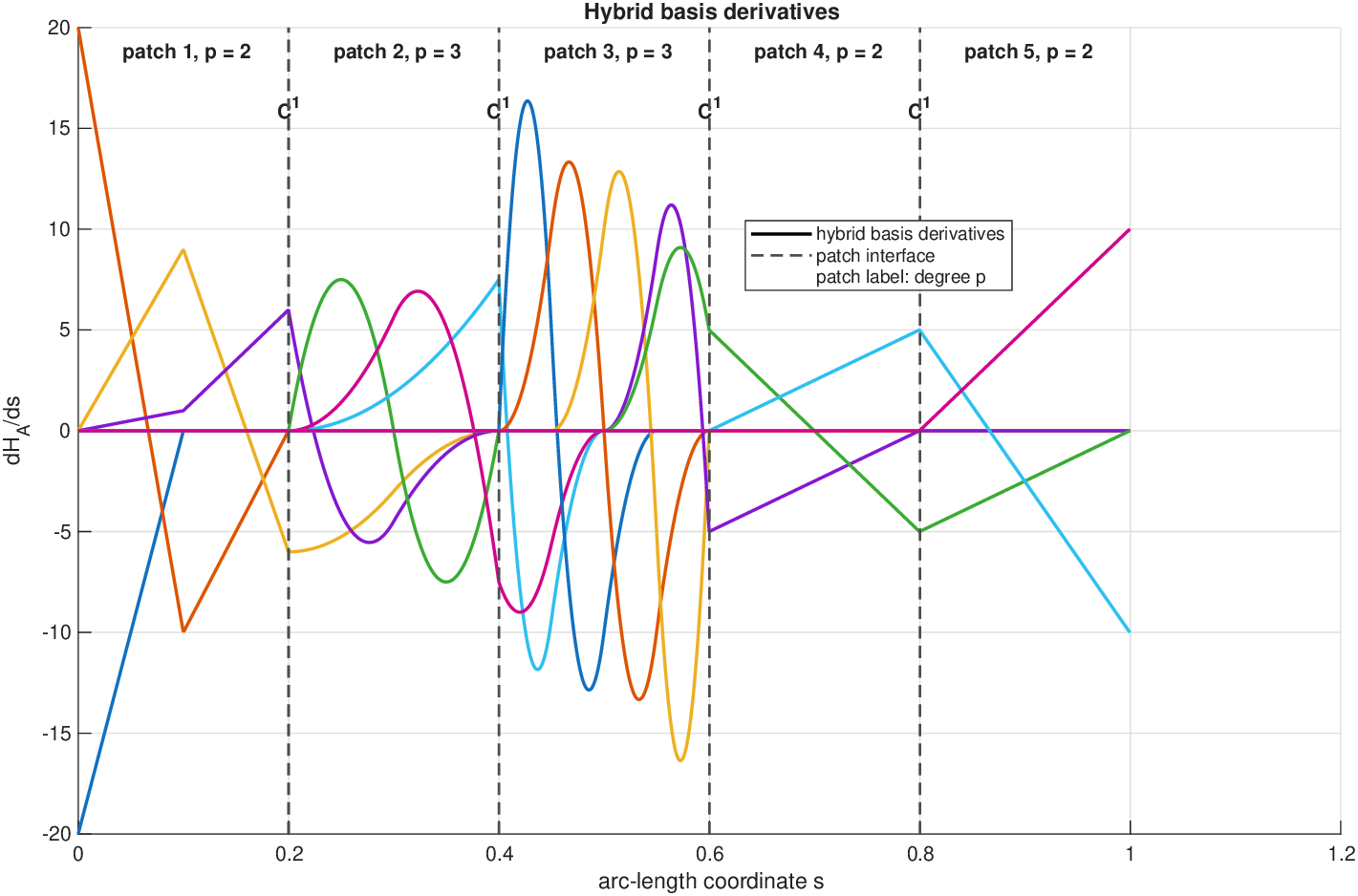} 
\caption{Continuity quality of derivatives (14 basis functions).}
\label{fig6}
\end{figure}

\paragraph{Result 3: Partition of unity.}
A plot of the sum of all basis functions over the domain (not shown) reveals that this sum is exactly $1$ everywhere (up to rounding error), confirming that the basis satisfies the partition of unity property. This is a necessary condition for invariance under constant fields and for proper interpolation.

\paragraph{Result 4: Conditioning of the basis.}
The hybrid basis matrix $T_{\text{hybrid}}$ has size $22 \times 14$ and full column rank ($\operatorname{rank}=14$). The condition number of $T_{\text{hybrid}}^T T_{\text{hybrid}}$ is $\kappa = 1.23\times 10^3$, which is moderate and indicates that the basis is well‑conditioned for numerical computations.

\paragraph{Result 5: Geometry reconstruction.}
In this paragraph we utilize the set of the 14 hybrid basis functions to investigate their capability of accurately representing a quarter-circle. The model starts from a rational quadratic B\'ezier patch (i.e., $p=2$) with the usual control points $P_0(1,0),P_1(1,1),P_2(0,1)$ and associated weights $w_0=1,w_1=\frac{\sqrt{2}}{2}$. Then, knot insertion and degree elevations according to the problem definition of Example~3 are performed, so that we eventually obtain the above-mentioned 14  hybrid basis, which now are directly related to the circular arc (see, Fig.~\ref{fig:circular_hybrid_basis}). In this procedure, the original curve is reconstructed by projecting its 14 control points onto the hybrid space. The location of the control points and the associated weights are shown in Table-3.
\begin{table}[htbp]
\centering
\caption{Cartesian coordinates and weights of 14 control points.}
\label{tab:control_points}
\begin{tabular}{c S[table-format=2.4] S[table-format=2.4] S[table-format=2.4]}
\toprule
{Point} & {\(x\)} & {\(y\)} & {\(w\)} \\
\midrule
     1  &  1.0000  &  0.0000  &  1.0000 \\
     2  &  1.0000  &  0.0728  &  0.9707 \\
     3  &  0.9077  &  0.4506  &  0.8708 \\
     4  &  0.9853  &  0.1980  &  0.9294 \\
     5  &  0.9010  &  0.4438  &  0.8750 \\
     6  &  0.7493  &  0.6715  &  0.8516 \\
     7  &  0.8768  &  0.4930  &  0.8672 \\
     8  &  0.7640  &  0.6470  &  0.8545 \\
     9  &  0.7079  &  0.7079  &  0.8531 \\
    10  &  0.6470  &  0.7640  &  0.8545 \\
    11  &  0.7176  &  0.7163  &  0.8476 \\
    12  &  0.4494  &  0.9078  &  0.8711 \\
    13  &  0.1502  &  1.0000  &  0.9414 \\
    14  &  0.0000  &  1.0000  &  1.0000 \\
\bottomrule
\end{tabular}
\end{table}

\begin{figure}[ht!]
\centering
\includegraphics[width=1.1\linewidth]{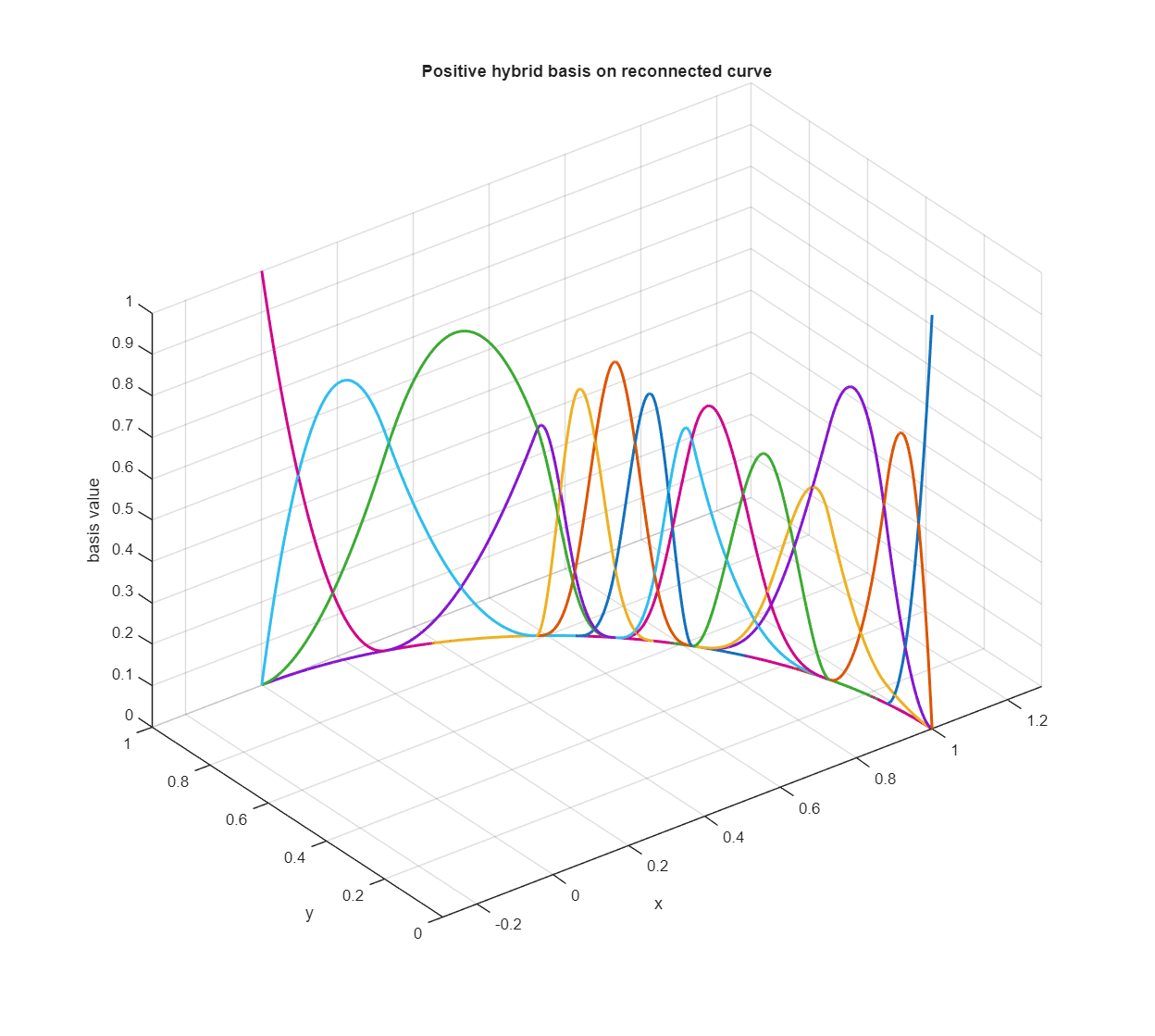} 
\caption{Hybrid basis functions on the circular arc (14 basis functions).}
\label{fig:circular_hybrid_basis}
\end{figure}

\bigskip

\bigskip

\noindent\textbf{EXAMPLE-4: A two‑patch, multi‑degree NURBS curve of a straight clamped beam}

This example shows the ability of the proposed hybrid basis to numerically solve a well-established boundary value problem of linear mechanics.

The hybrid basis is used to solve a horizontal cantilever beam of length ($L=1$) with a point load at $x=a=0.65$. The material properties are $EA=10^6$ and $EI=10^2$. The boundary conditions are clamped at the left end ($u_x=u_y=\theta=0$), and a vertical load $P=-10$ is applied at the interior point. The numerical solution is compared with the analytical Euler–Bernoulli solution for a straight beam. The governing equation is
\begin{equation}
    E I \frac{d^4 w}{dx} = P \delta(x-a) \,,
\end{equation}
where $\delta$ is Dirac delta function, and the weak IGA formulation leads to elements $k_{ij}$ of the stiffness matrix given --in terms of univariate NURBS $N_k$-- by
\begin{equation}
    k_{ij} = \int_0^{L} \frac{d^2 N_i}{dx^2} \frac{d^2 N_j}{dx^2} \, dx \,.
\end{equation}

The analytical solution is a piecewise-defined function, as follows:
\begin{equation}
u(x) =
\begin{cases}
u_{1}(x) = \dfrac{P x^{2}}{6EI}(3a - x), & 0 \le x \le a \quad\text{(before the load)},\\[1.2ex]
u_{2}(x) = \dfrac{P a^{2}}{6EI}(3x - a), & a \le x \le L \quad\text{(after the load)}.
\end{cases} \label{eq39}
\end{equation}

Equation~\eqref{eq39} suggests that the optimal hybrid model to approximate the above situation is an assemblage of two B-spline elements ($[0,a],[a,1]$) with $p_1=3,p_2=1$. The latter model includes four basis functions; these and their associated derivatives are shown in Fig.~\ref{fig8} and Fig.~\ref{fig9}, respectively.

\begin{figure}[ht!]
\centering
\includegraphics[width=1.0\linewidth]{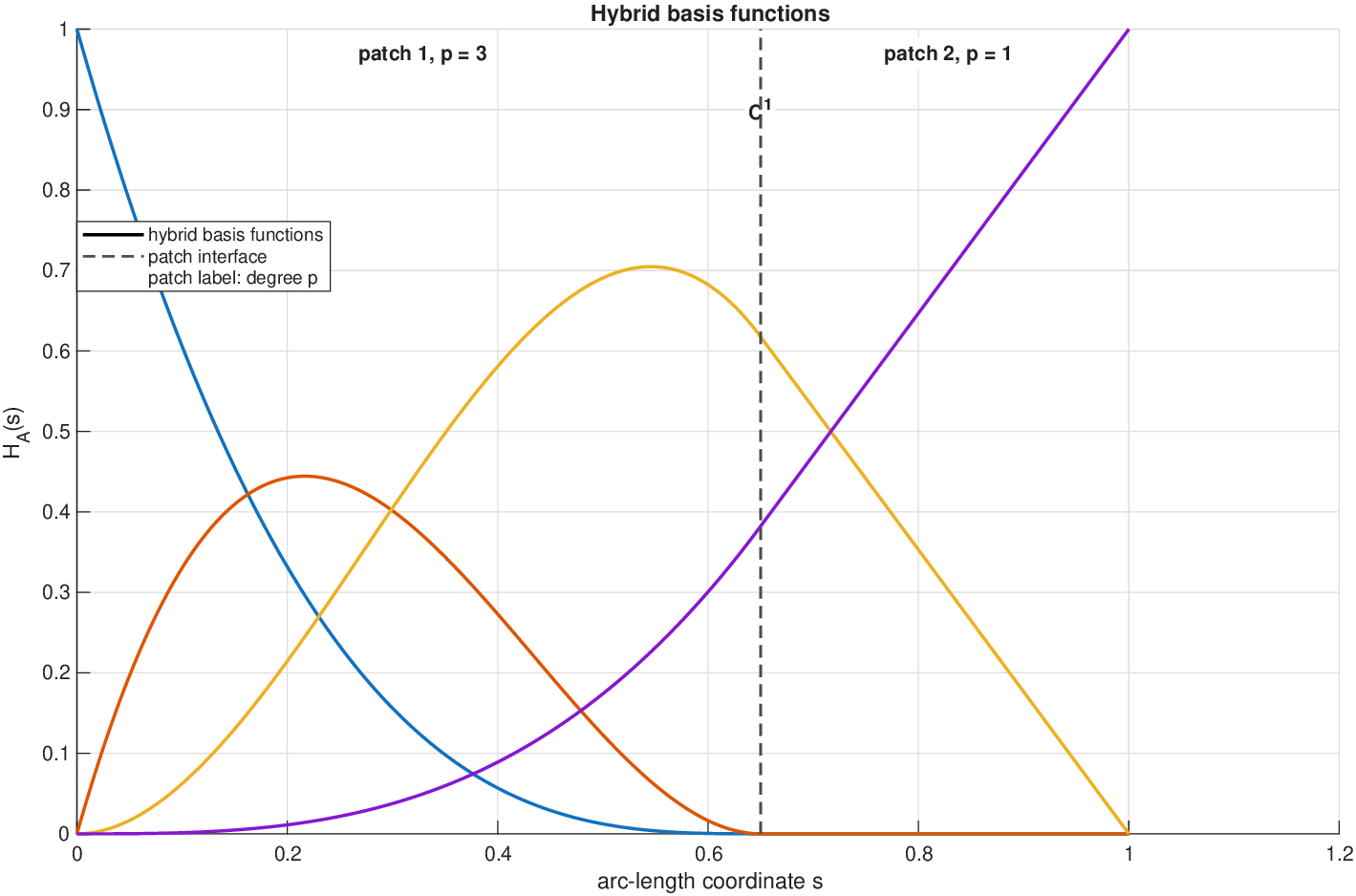} 
\caption{Hybrid basis functions on the cantilever beam (4 basis functions).}
\label{fig8}
\end{figure}
\begin{figure}[h!]
\centering
\includegraphics[width=1.0\linewidth]{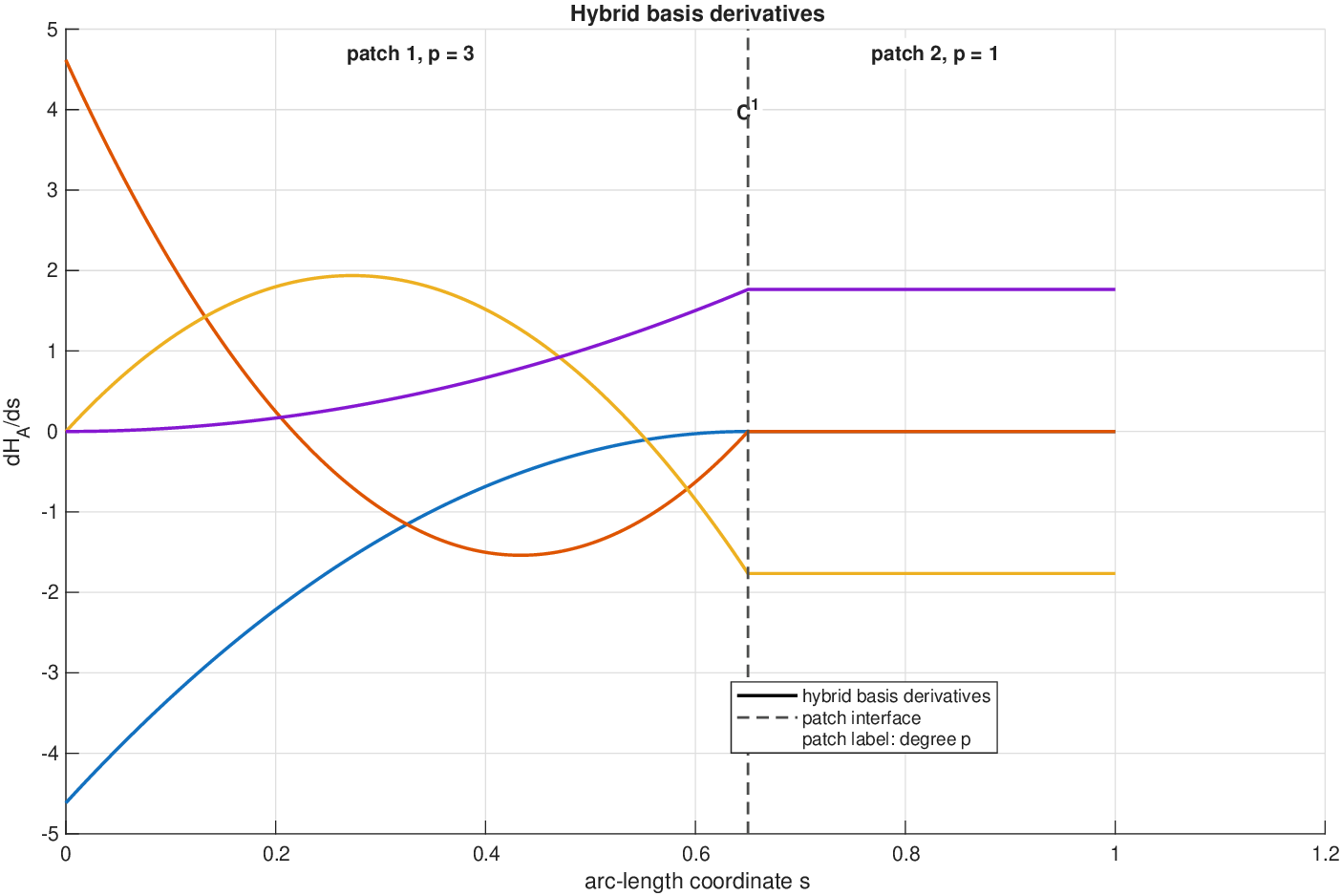} 
\caption{Derivatives of Hybrid basis functions on the cantilever beam (4 basis functions).}
\label{fig9}
\end{figure}   

Of course, it is possible to solve the same problem applying higher degrees as well. For example, the case $p_1=4, p_2=2$ leads to 14 basis functions, which are presented in Fig.~\ref{fig10} whereas their derivatives in Fig.~\ref{fig11}.
\begin{figure}[ht!]
\centering
\includegraphics[width=1.0\linewidth]{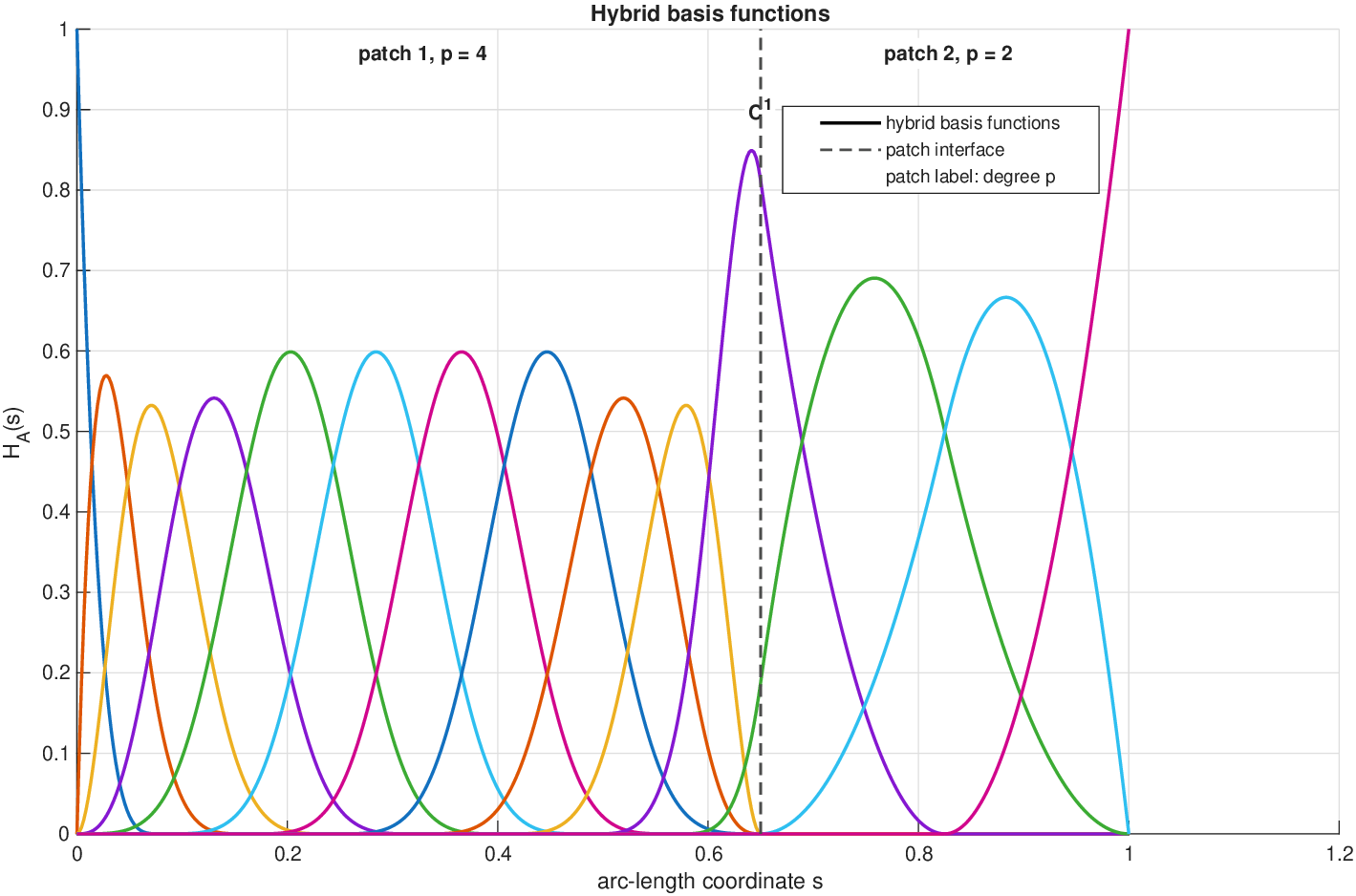} 
\caption{Hybrid basis functions on the cantilever beam (4 basis functions).}
\label{fig10}
\end{figure}
\begin{figure}[h!]
\centering
\includegraphics[width=1.0\linewidth]{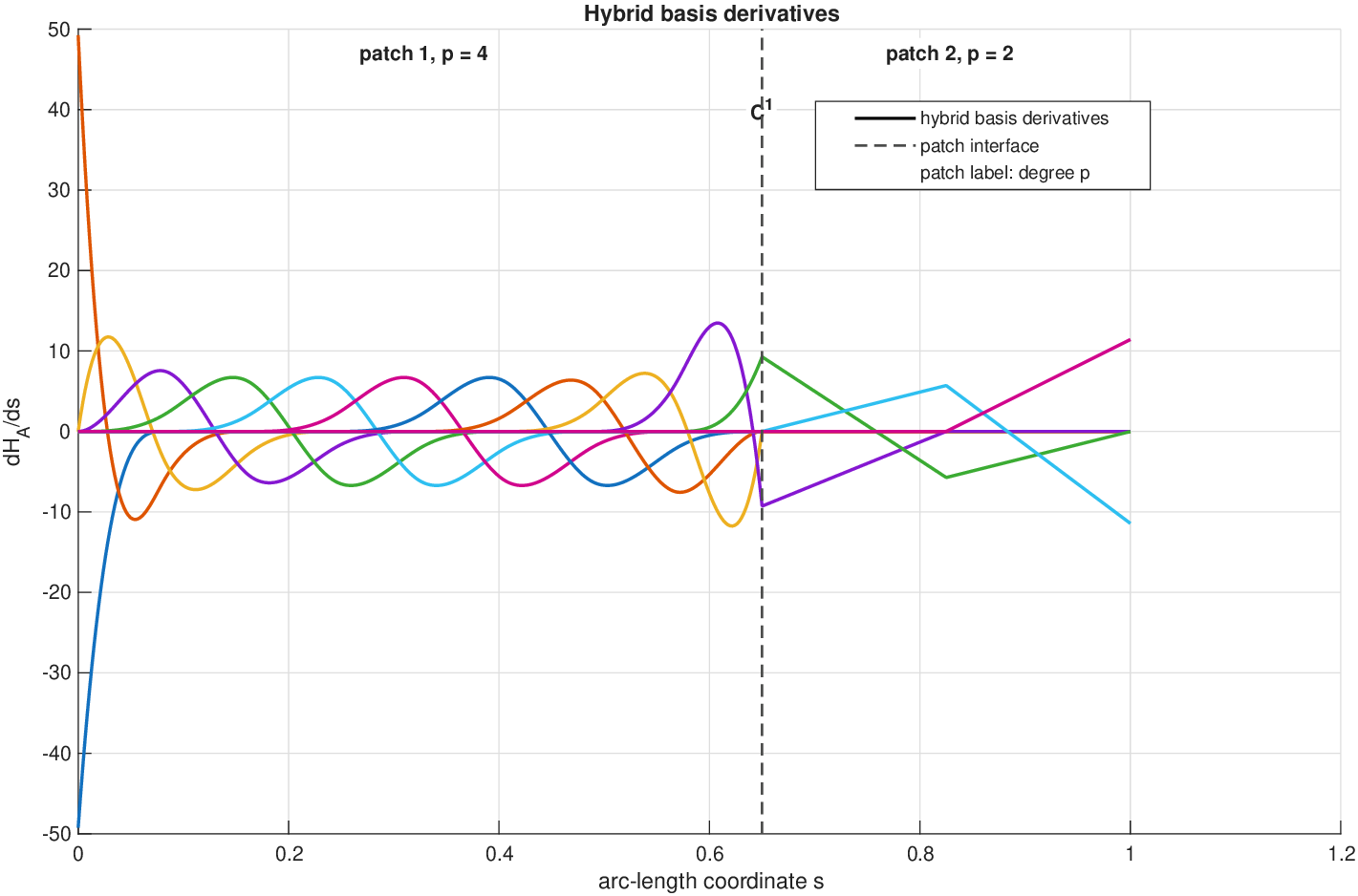} 
\caption{Derivatives of Hybrid basis functions on the cantilever beam (4 basis functions).}
\label{fig11}
\end{figure}   

The corresponding error norms are $L_2=4.82 \times 10^{-8}$ and $L_2=6.41 \times 10^{-8}$, which indicates convergence when tending to the accurate choice $p_1=3,p_2=1$. The nonzero value is attributed to truncation errors.


\bigskip

\noindent\textbf{EXAMPLE 5: Verification of the hybrid Isogeometric formulation under local $p$-refinement}

The objective of this example is to verify the accuracy and convergence
properties of the proposed hybrid Isogeometric formulation. Unlike a
classical convergence study based on uniform mesh refinement, the present
test investigates the capability of the hybrid basis to reproduce the
exact solution through local polynomial enrichment while maintaining a
fixed geometric discretization.

The computational domain consists of two NURBS patches joined through 
a hybrid coupling interface. The geometry is kept unchanged throughout 
the analysis; therefore no $h$-refinement is performed. 
Instead, the local approximation space is enriched by alternately 
increasing the polynomial degree of one patch at a time. 
Consequently, the sequence of approximation spaces is

\[
(2,2)
\rightarrow
(3,2)
\rightarrow
(3,3)
\rightarrow
(4,3)
\rightarrow
(4,4)
\rightarrow
(5,4),
\]

where each ordered pair denotes the polynomial degree of the left and
right patches, respectively.

The purpose of adopting local $p$-refinement instead of $h$-refinement
is twofold.

First, the exact solution is smooth inside each patch but possesses
different polynomial orders in the two subdomains. Therefore, increasing
the local approximation order provides the most direct way to investigate
whether the hybrid basis reproduces the exact polynomial space without
introducing unnecessary mesh refinement.

Second, keeping the mesh fixed isolates the influence of the hybrid
coupling itself. Since no additional elements are introduced during the
analysis, any reduction of the discretization error can be attributed
solely to the enrichment of the approximation space and to the exact
coupling enforced by the hybrid basis functions.

Uniform $h$-refinement is therefore intentionally omitted in this
verification example. A dedicated $h$- or $hp$-adaptivity study is
presented separately.

The nonlinear boundary value problem considered is

\begin{equation}
-
\frac{d}{dx}
\left(
u\frac{du}{dx}
\right)
=
f(x),
\qquad
x\in(0,1),
\end{equation}

subject to Dirichlet boundary conditions obtained from the exact
manufactured solution

\begin{equation}
u(x)=
\begin{cases}
1+x^5,
&
0\le x\le a,
\\[2mm]
1+\dfrac{5a}{4}x^4-\dfrac{a^5}{4},
&
a<x\le1,
\end{cases}
\end{equation}

where the interface is located at

\[
a=0.5.
\]

The forcing function $f(x)$ is obtained analytically by substituting the
exact solution into the governing differential equation, ensuring that
the exact solution satisfies the nonlinear problem identically.

This manufactured solution is particularly suitable for assessing local
polynomial adaptivity because the polynomial order differs between the two
patches. The left patch requires a fifth-order polynomial representation,
whereas the right patch requires only a fourth-order polynomial.
Consequently, the exact solution belongs to the discrete approximation
space as soon as the local polynomial degrees become

\[
(p_1,p_2)=(5,4).
\]

At convergence, the hybrid basis functions are shown in Fig.~\ref{fig12}. 
%
\begin{figure}[h!]
\centering
\includegraphics[width=1.0\linewidth]{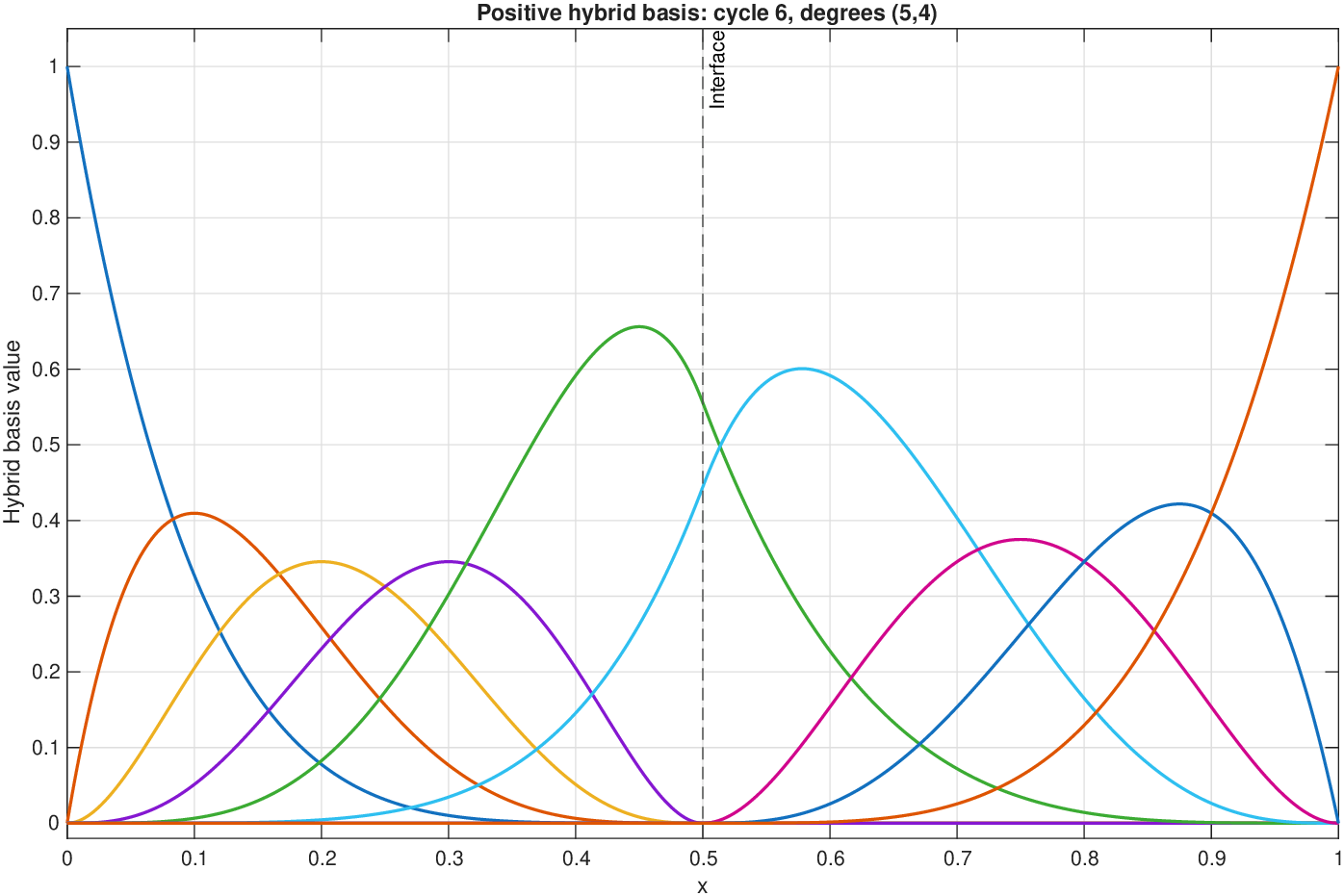} 
\caption{Hybrid basis functions at convergence $(p_1=5,p_2=4)$.}
\label{fig12}
\end{figure}   

From this point onward, the discretization error is expected to decrease
to machine precision, thereby providing a stringent verification of the
hybrid basis construction, the interface coupling, and the nonlinear
solver.

Actually, convergence quality of the proposed method versus the standard uniform degree elevation approach (the latter implemented into GeoPDEs), is shown in Fig.~\ref{fig13}, where one may observe that the proposed method performs well.
\begin{figure}[h!]
\centering
\includegraphics[width=1.0\linewidth]{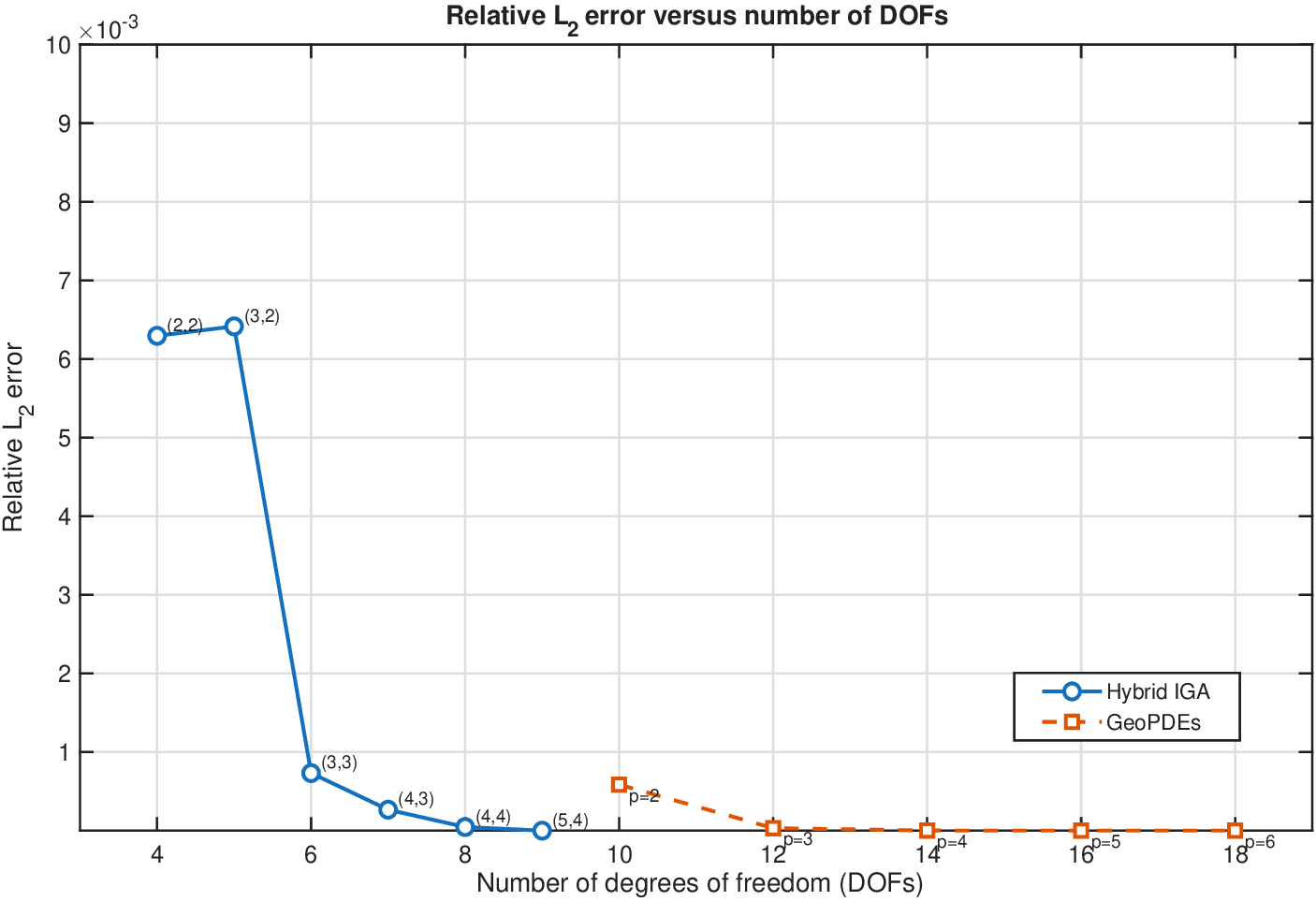} 
\caption{Convergence of the proposed Hybrid IGA method versus the uniform degree elevation method (GeoPDEs).}
\label{fig13}
\end{figure}   






\section{Discussion}

\textcolor{black}{The algebraic merging algorithm presented in this work is completely independent of the local basis type. In Example 2, the local patches were \textit{Bernstein} polynomials; in Example 3, they are \textit{unclamped} NURBS. In both cases, the same procedure is applied: build the continuity constraints, solve for a non‑negative basis of the null space, and enforce the partition of unity via non‑negative least squares. The only input required is the evaluation of the local basis functions and their derivatives at the interfaces. This makes the method \textit{universally applicable} to any spline or polynomial representation that admits a local basis.}

In Example~2 of Section~\ref{sec:validation}, we demonstrate that the proposed algebraic construction reproduces the classical B\'ezier extraction operator for a two-element quadratic spline.

Unlike the pedagogical Example~2, which used MATLAB's \texttt{null} and ray enumeration, the proposed general method employs a column-pivoted QR factorization followed by linear programming to generate non-negative null-space vectors, and non-negative least squares for the partition-of-unity scaling. This approach is computationally more efficient, avoids dense matrices, and scales linearly with the number of degrees of freedom.

The proposed methodology should not be interpreted as a new spline technology in the traditional sense. Rather, it provides a reconstruction framework that allows existing spline representations to be decomposed, manipulated locally, and subsequently reconstructed into admissible analysis spaces. The framework therefore focuses on the relationship between local spline operations and global admissibility rather than on the definition of a new basis family. In this section, we discuss several important aspects, implications, and limitations of the proposed approach.

\subsection{Adaptive Characteristics}

The decomposition procedure naturally supports local modifications since each Active Section possesses its own local spline description. As established in Section 4.2, an Active Section contains exactly \(2p+2\) knots and encapsulates all basis functions non-zero on a given knot span. Local knot insertion may be introduced in selected regions without directly modifying neighboring Active Sections, since the outermost \(p\) knots of each section act as buffers that isolate the internal spline data from adjacent regions. Similarly, degree elevation may be performed independently within individual sections. In this sense, the framework exhibits adaptive characteristics.

However, the present work does not introduce a complete adaptive refinement strategy driven by error indicators or solution-based refinement criteria. Instead, the focus is placed on the reconstruction mechanism required after local modifications have been performed. The development of fully adaptive refinement procedures, including error estimation and refinement indicators, remains a topic for future investigation. Nevertheless, the framework provides a flexible foundation for adaptive strategies: since local modifications are entirely independent, refinement criteria can be evaluated locally and applied directly to the corresponding Active Sections, with the reconstruction stage ensuring global admissibility.

\subsection{Relationship to Hierarchical Methods}

The proposed methodology is not inherently hierarchical. Modern hierarchical constructions such as THB-splines~\cite{giannelli2012thb} and more recent adaptive spline technologies achieve locality through hierarchical activation and refinement mechanisms, where admissibility is maintained throughout the refinement process. In contrast, the present framework operates through decomposition and reconstruction: local spline entities are treated independently and continuity is recovered through algebraic constraints.

Nevertheless, hierarchical refinement strategies could potentially be incorporated within individual Active Sections prior to reconstruction. For example, a hierarchical basis could be constructed within a selected Active Section using standard hierarchical refinement techniques, and the resulting locally refined section could then be integrated into the global reconstruction process. Consequently, the proposed framework should be viewed as complementary to hierarchical methodologies rather than as a replacement for them.

Unlike classical macro-element approaches, the local entities employed here are obtained directly from an exact NURBS decomposition. Each Active Section is an unclamped spline entity, meaning its endpoint knots do not have multiplicity \(p+1\). This choice, discussed in Section 4.2, is deliberate: it preserves the native spline character of the local representation and permits unrestricted local modifications within the interior of the section, while the buffer knots retain the information required for subsequent reconstruction.

\subsection{Compatibility with Existing Spline Technologies}

The construction of analysis-suitable spline spaces across multiple patches has received considerable attention in the isogeometric analysis literature. In particular, smooth multi-patch spline spaces, analysis-suitable parameterizations, and continuity-preserving constructions have been investigated extensively~\cite{collin2016analysis, hughessangallitakacstoshniwal2021}. The proposed reconstruction framework does not seek to replace these technologies. Instead, it provides an alternative viewpoint in which admissibility is recovered after local spline manipulation rather than maintained continuously during space construction.

The reconstruction procedure is largely independent of the specific spline representation employed within each Active Section. Although the present work focuses on NURBS-based descriptions, the underlying reconstruction philosophy is not restricted to a particular spline technology. Potential extensions may include hierarchical spline representations~\cite{giannelli2012thb}, truncated hierarchical splines (THB-splines)~\cite{giannelli2012thb}, locally refined splines (LR-splines)~\cite{dokken2013locally}, extraction-based spline technologies~\cite{borden2011isogeometric, scott2012analysis}, U-splines~\cite{herrema2018usplines}, and multi-patch spline descriptions~\cite{takacs2025multiresolution, collin2016analysis}. The primary requirement is the ability to construct continuity constraints between neighboring local entities. Once such constraints can be expressed algebraically, the reconstruction framework may be applied independently of the specific spline technology used within the individual Active Sections.

\subsection{Alternative Reconstruction Strategies}

The reconstruction procedure presented in this work relies on the null space of the global constraint operator. This choice was motivated by the ability of null-space operators to generate admissible spaces satisfying the prescribed continuity requirements while offering precise control over the properties of the resulting basis (positivity, locality, and partition of unity). However, the framework itself is not fundamentally tied to a null-space formulation.

Alternative reconstruction procedures may be considered, including direct constraint elimination, optimization-based reconstruction, penalty formulations, mortar-type approaches, and weak coupling methods. The null-space approach should therefore be interpreted as one possible realization of the broader reconstruction framework. The linear programming formulation employed in the present work offers a particular balance between computational cost and the ability to enforce strict non-negativity and locality. For problems with a large number of degrees of freedom, alternative strategies such as constraint elimination or reduced-order modeling may prove more efficient.

\subsection{Local Flexibility versus Global Admissibility}

A natural question concerns whether the global nature of the Hybrid Reconstruction Operator eliminates the locality introduced by the Active Section framework. At first sight, the reconstruction process appears to replace a collection of local spline entities by a global algebraic operator, \(\mathbf{T}_{\text{hybrid}}\). Such an interpretation would suggest that the advantages gained during decomposition are ultimately lost during reconstruction. However, this is not the case.

The proposed framework separates two distinct tasks that are traditionally coupled in spline-based analysis. The first task concerns local spline manipulation, including knot insertion, degree elevation, basis modification, and local refinement. These operations are performed independently within each Active Section—confined to the interior interval between the buffer knots—and do not require the immediate maintenance of global continuity constraints. The second task concerns the recovery of an admissible approximation space. This is achieved only after the local modifications have been completed, through the construction of the Hybrid Reconstruction Operator.

Consequently, locality and admissibility are not competing objectives but operate at different stages of the methodology. Locality is exploited during geometric and spline manipulations, whereas admissibility is recovered afterwards through algebraic reconstruction. The numerical experiments presented in this work support this viewpoint. As the number of Active Sections increases, the reconstructed spaces retain a substantial portion of the local approximation freedom while satisfying all continuity requirements. For the examples considered, approximately seventy percent of the disconnected local degrees of freedom remain available after admissibility reconstruction. Furthermore, the resulting Hybrid Reconstruction Operators remain well conditioned, while partition of unity, geometric exactness, and continuity are preserved to machine precision.

These observations suggest that the principal advantage of the proposed framework is not the elimination of global admissibility, but rather the decoupling of local spline manipulation from continuity enforcement. Local modifications may therefore be performed with a high degree of flexibility, while globally consistent approximation spaces are recovered only when required for analysis.

\subsection{Clamped versus Unclamped Reconstruction}

Two reconstruction strategies were investigated during the development of the present framework. The first strategy converts each Active Section into a clamped local representation prior to reconstruction. In this setting, interface degrees of freedom become explicitly identifiable and continuity constraints may be imposed in a relatively straightforward manner. The second strategy retains the original unclamped Active Sections and performs admissibility recovery directly within the resulting spline spaces. This formulation is considerably more challenging because the local spline entities do not possess the endpoint interpolation properties associated with clamped representations.

The present work adopts the unclamped formulation, as defined in Section 4.2. This choice is motivated by several considerations. First, it preserves the native spline character of the local representation throughout the decomposition and reconstruction processes. Second, unclamped Active Sections provide greater flexibility for local knot insertion and degree elevation, since the interior knots are not constrained by endpoint multiplicities. Third, the buffer knots—the outermost \(p\) knots on each side—retain all information required for subsequent reconstruction, ensuring that continuity constraints can still be imposed even though the representation is unclamped. Although the implementation is more involved, the unclamped formulation offers a cleaner separation between local modification and global reconstruction.

\subsection{Current Limitations}

The numerical developments presented in this work demonstrate the effectiveness of the methodology for one-dimensional spline representations and associated analysis problems. Extension to multidimensional configurations introduces additional challenges related to interface topology, continuity enforcement, and local refinement compatibility. Preliminary two-dimensional investigations indicate that the reconstruction philosophy remains applicable, although robust local \(h\)- and \(p\)-modification strategies require further development.

For this reason, the present work focuses primarily on the one-dimensional setting, where the reconstruction process can be examined in a clear and controlled manner. The extension to surfaces and volumes, including the treatment of T-junctions and unstructured multi-patch topologies, is left for future work. Additionally, the computational cost of the linear programming formulation may become significant for problems with a very large number of active constraints, suggesting the need for more efficient optimization strategies or reduced-order approximations in such cases.

Finally, while the framework guarantees the theoretical properties of the reconstructed basis (positivity, partition of unity, and admissibility), the conditioning of the resulting basis and the stability of the reconstruction procedure for very high degrees or heavily refined meshes warrant further investigation. These aspects will be addressed in future work.

\section{Conclusions}

We have presented a reconstruction-oriented framework that decouples local spline manipulation from global admissibility enforcement in the context of multipatch isogeometric analysis. By decomposing a NURBS representation into independent unclamped Active Sections—each defined by a minimal local knot window of \(2p+2\) knots that completely determines the \(p+1\) basis functions on a given span—we enable arbitrary local knot insertion, degree elevation, and basis modification without affecting neighboring regions or compromising exact CAD geometry. The subsequent reconstruction stage recovers global admissibility through the assembly of interface continuity constraints into a sparse global operator and the construction of a Hybrid Reconstruction Operator, which, via anchor selection, linear programming with distance-based regularization, and non-negative least-squares, yields a strictly non-negative, partition-of-unity basis that spans the exact constrained nullspace to machine precision. An efficient hierarchical assembly strategy, which freezes columns already satisfying new interface constraints, ensures that the optimization burden is confined exclusively to active degrees of freedom at each pairwise merge, making the framework computationally scalable. Numerical benchmarks on one-dimensional multipatch curves, including geometry preservation tests, continuity recovery, local degree variation, and a nonlinear diffusion problem with heterogeneous local polynomial degrees, confirm optimal convergence rates and demonstrate the methodology's ability to seamlessly combine local geometric flexibility with globally consistent approximation spaces. The proposed framework is not a new spline technology but rather a generic algebraic procedure that can be superimposed on any spline representation—hierarchical, T-spline, LR-spline, or U-spline—capable of expressing continuity constraints algebraically. Future work will focus on extension to multidimensional configurations, development of fully adaptive refinement strategies driven by error indicators, and investigation of the conditioning and stability of the reconstructed basis for very high degrees and heavily refined meshes.

\bibliographystyle{unsrt} 
\bibliography{references}

@article{hughes2005isogeometric,
  title={Isogeometric analysis: CAD, finite elements, NURBS, exact geometry and mesh refinement},
  author={Hughes, Thomas J.R. and Cottrell, John A. and Bazilevs, Yuri},
  journal={Computer Methods in Applied Mechanics and Engineering},
  volume={194},
  number={39--41},
  pages={4135--4195},
  year={2005}
}

@book{piegl1997nurbs,
  title={The NURBS Book},
  author={Piegl, Les and Tiller, Wayne},
  year={1997},
  publisher={Springer}
}

@article{vuong2011hierarchical,
  title={A hierarchical approach to adaptive local refinement in isogeometric analysis},
  author={Vuong, An Nguyen and Giannelli, Carlotta and Juettler, Bert and Simeon, Bernd},
  journal={Computer Methods in Applied Mechanics and Engineering},
  volume={200},
  pages={3554--3567},
  year={2011}
}

@article{bazilevs2006isogeometric,
  title={Isogeometric analysis using T-splines},
  author={Bazilevs, Yuri and Calo, Victor M. and Cottrell, John A. and Hughes, Thomas J.R.},
  journal={Computer Methods in Applied Mechanics and Engineering},
  volume={199},
  pages={229--263},
  year={2006}
}

@article{scott2011isogeometric,
  title={Isogeometric analysis using unstructured T-splines},
  author={Scott, Michael A. and Li, Xianfeng and Sederberg, Thomas W. and Hughes, Thomas J.R.},
  journal={Computer Methods in Applied Mechanics and Engineering},
  volume={200},
  pages={2297--2310},
  year={2011}
}

@book{cottrell2009isogeometric,
  title={Isogeometric analysis: Toward integration of CAD and FEA},
  author={Cottrell, John A. and Hughes, Thomas J.R. and Bazilevs, Yuri},
  year={2009},
  publisher={Wiley}
}

@article{herrema2018usplines,
  title={U-splines: Splines for unstructured meshes},
  author={Herrema, Andrew J. and Scott, Michael A. and Evans, John A.},
  journal={Computer Methods in Applied Mechanics and Engineering},
  volume={327},
  pages={325--354},
  year={2018}
}

@article{herrema2017adaptive,
  title={Adaptive isogeometric analysis with U-splines},
  author={Herrema, Andrew J. and Evans, John A.},
  journal={Computer Methods in Applied Mechanics and Engineering},
  volume={316},
  pages={966--1000},
  year={2017}
}

@article{wohlmuth2001mortar,
  title={A mortar finite element method using dual spaces for the Lagrange multiplier},
  author={Wohlmuth, Barbara},
  journal={SIAM Journal on Numerical Analysis},
  volume={38},
  pages={989--1012},
  year={2001}
}

@book{zienkiewicz2005finite,
  title={The Finite Element Method},
  author={Zienkiewicz, O.C. and Taylor, R.L.},
  year={2005},
  publisher={Butterworth-Heinemann}
}

@article{hansbo2005nitsche,
  title={Nitsche's method for interface problems in finite elements},
  author={Hansbo, Peter and Hansbo, Anita},
  journal={Computer Methods in Applied Mechanics and Engineering},
  volume={193},
  pages={4195--4207},
  year={2005}
}

@article{bernardi1993domain,
  title={A domain decomposition method for elliptic problems with nonmatching grids},
  author={Bernardi, Christine and Maday, Yvon and Patera, Anthony T.},
  journal={SIAM Journal on Numerical Analysis},
  volume={30},
  pages={152--169},
  year={1993}
}

@book{quarteroni1999domain,
  title={Domain Decomposition Methods for Partial Differential Equations},
  author={Quarteroni, Alfio and Valli, Alberto},
  year={1999},
  publisher={Oxford University Press}
}

@article{coox2017robust,
  title={A robust patch coupling method for NURBS-based isogeometric analysis of non-conforming multipatch surfaces},
  author={Coox, Laurens and Greco, Francesco and Atak, Onur and Vandepitte, Dirk and Desmet, Wim},
  journal={Computer Methods in Applied Mechanics and Engineering},
  volume={316},
  pages={986--1017},
  year={2017}
}

@article{evans2009nitsche,
  title={Nitsche’s method for two and three dimensional NURBS patch coupling},
  author={Evans, John A and Hughes, Thomas JR},
  journal={Computer Methods in Applied Mechanics and Engineering},
  volume={199},
  number={9-12},
  pages={671--684},
  year={2009}
}

@article{giannelli2012thb,
  title={THB-splines: The truncated basis for hierarchical splines},
  author={Giannelli, Carlotta and J{\"u}ttler, Bert and Speleers, Hendrik},
  journal={Computer Aided Geometric Design},
  volume={29},
  number={7},
  pages={485--498},
  year={2012}
}

@article{dokken2013locally,
  title={Locally refined splines},
  author={Dokken, Tor and Lyche, Tom and Pettersen, Knut M{\o}rken},
  journal={Computer Aided Geometric Design},
  volume={30},
  number={3},
  pages={331--356},
  year={2013}
}

@article{popp2012dual,
  title={Dual mortar methods for isogeometric analysis},
  author={Popp, Alexander and Gee, Michael W and Wall, Wolfgang A},
  journal={Computer Methods in Applied Mechanics and Engineering},
  volume={245},
  pages={273--290},
  year={2012}
}

@article{apostolatos2014nitsche,
  title={Nitsche’s method for coupling non-matching isogeometric shells},
  author={Apostolatos, Alexandros and Schmidt, Roland and W{\"u}chner, Roland and Bletzinger, Kai-Uwe},
  journal={Computer Methods in Applied Mechanics and Engineering},
  volume={284},
  pages={673--699},
  year={2014}
}

@article{ruess2014weak,
  title={Weak coupling for isogeometric analysis of non-matching and trimmed multi-patch geometries},
  author={Ruess, Martin and Schillinger, Dominik and Bazilevs, Yuri and Varduhn, Volker and Rank, Ernst},
  journal={Computer Methods in Applied Mechanics and Engineering},
  volume={269},
  pages={46--71},
  year={2014}
}

@article{takacs2025multiresolution,
  author     = {Takacs, Stefan and Tyoler, Stefan},
  title      = {Multi-resolution isogeometric analysis -- efficient adaptivity utilizing the multi-patch structure},
  journal    = {Computers \& Mathematics with Applications},
  volume     = {179},
  pages      = {103--125},
  year      = {2025},
  doi       = {10.1016/j.camwa.2024.12.005},
  publisher = {Elsevier}
}

@article{sederberg2003tsplines,
  title={T-splines and T-NURCCs},
  author={Sederberg, Thomas W and Zheng, Jianmin and Bakenov, Almaz and Nasri, Ahmad},
  journal={ACM Transactions on Graphics},
  volume={22},
  number={3},
  pages={477--484},
  year={2003},
  publisher={ACM}
}

@article{scott2012analysis,
  title={Isogeometric finite element data structures based on Bézier extraction of T-splines},
  author={Scott, Michael A and Evans, John A and Borden, Michael J and Hughes, Thomas JR},
  journal={International Journal for Numerical Methods in Engineering},
  volume={88},
  number={2},
  pages={126--156},
  year={2011},
  publisher={Wiley Online Library}
}

@article{borden2011isogeometric,
  title={Isogeometric finite element data structures based on Bézier extraction},
  author={Borden, Michael J and Scott, Michael A and Evans, John A and Hughes, Thomas JR},
  journal={International Journal for Numerical Methods in Engineering},
  volume={87},
  number={1-5},
  pages={15--47},
  year={2011},
  publisher={Wiley Online Library}
}

@incollection{HughesSangalliTakacsToshniwal2021,
  author    = {Thomas J. R. Hughes and Giancarlo Sangalli and Thomas Takacs and Deepesh Toshniwal},
  title     = {Smooth Multi-Patch Discretizations in Isogeometric Analysis},
  booktitle = {Geometric Partial Differential Equations -- Part II},
  series    = {Handbook of Numerical Analysis},
  volume     = {22},
  pages      = {467--543},
  year       = {2021},
  publisher  = {Elsevier},
  doi        = {10.1016/bs.hna.2020.09.002}
}

@article{collin2016analysis,
  author  = {Collin, Adrien and Sangalli, Giancarlo and Takacs, Thomas},
  title   = {Analysis-Suitable G1 Multi-Patch Parametrizations for C1 Isogeometric Spaces},
  journal = {Computer Aided Geometric Design},
  volume  = {47},
  pages   = {93--113},
  year    = {2016}
}

@article{kapl2017isogeometric,
  author    = {Kapl, Mario and Buchegger, Florian and Bercovier, Michel and J{\"u}ttler, Bert},
  title     = {Isogeometric analysis with geometrically continuous functions on planar multi-patch geometries},
  journal   = {Computer Methods in Applied Mechanics and Engineering},
  volume    = {316},
  pages     = {209--234},
  year      = {2017},
  publisher = {Elsevier}
}

\qquad
\qquad

\begin{center}
{\Large\bfseries
APPENDICES }
\end{center}

\appendix

\section{  Mathematical Foundations of the Hybrid Reconstruction    }
\label{app:A}

\addcontentsline{toc}{section}
{Appendix A. Mathematical Foundations of the Hybrid Reconstruction}

\setcounter{subsection}{0}
\renewcommand{\thesubsection}{A.\arabic{subsection}}

\setcounter{subsubsection}{0}
\renewcommand{\thesubsubsection}{\thesubsection.\arabic{subsubsection}}

\subsection{Local Unclamped Decomposition and Pairwise Hybrid Reconstruction}
\label{app:local_pairwise}

This appendix establishes the mathematical basis of the local
unclamped decomposition and the subsequent recursive pairwise
reconstruction. Particular attention is given to the distinction
between:

\begin{enumerate}[label=(\roman*)]
    \item the local knot vector associated with a single Active Section;
    \item the disconnected approximation space formed by two neighboring
    Active Sections before interface coupling.
\end{enumerate}

Although both constructions involve the quantity \(2p+2\), they describe
different mathematical objects.

\subsection{Active Sections and the \(2p+2\) local knot vector}
\label{app:active_section_2p2}

Let
\[
\Xi
=
\{\xi_0,\xi_1,\ldots,\xi_{n+p+1}\}
\]
be the knot vector of a univariate B-spline or NURBS representation of
degree \(p\). Consider a nonzero knot span
\[
I_e=[\xi_e,\xi_{e+1}],
\qquad
\xi_e<\xi_{e+1}.
\]

Exactly \(p+1\) degree-\(p\) B-spline basis functions are nonzero over the
interior of \(I_e\). Their indices are
\[
e-p,\ldots,e.
\]

\begin{definition}[Active Section]
\label{def:active_section}
The Active Section associated with the nonzero knot span
\(I_e=[\xi_e,\xi_{e+1}]\) is the local spline entity obtained by
restricting to \(I_e\) the \(p+1\) parent basis functions that are
nonzero on that span, together with their corresponding control points,
weights, and inherited local knot data.
\end{definition}

Thus, an Active Section is associated with one element, and not with the
union of two elements located on the two sides of an interface.

\begin{proposition}[Local knot-vector length of an Active Section]
\label{prop:active_section_knots}
The Active Section associated with \(I_e\) is characterized by the local
knot subsequence
\[
\Xi_e^{\mathrm{loc}}
=
\{
\xi_{e-p},
\xi_{e-p+1},
\ldots,
\xi_{e+p+1}
\}.
\]
This local knot vector contains exactly
\[
\boxed{2p+2}
\]
knot entries.
\end{proposition}

\begin{proof}
The basis functions active on \(I_e\) are
\[
N_{e-p,p},
N_{e-p+1,p},
\ldots,
N_{e,p}.
\]

A degree-\(p\) B-spline basis function \(N_{i,p}\) is defined by the knot
subsequence
\[
\{\xi_i,\xi_{i+1},\ldots,\xi_{i+p+1}\}.
\]

The first active basis function, \(N_{e-p,p}\), begins at
\(\xi_{e-p}\), whereas the last active basis function, \(N_{e,p}\),
ends at \(\xi_{e+p+1}\). Therefore, the smallest knot subsequence
containing the complete knot data required by all basis functions active
on \(I_e\) is
\[
\Xi_e^{\mathrm{loc}}
=
\{
\xi_{e-p},
\xi_{e-p+1},
\ldots,
\xi_{e+p+1}
\}.
\]

The number of entries in this sequence is
\[
(e+p+1)-(e-p)+1
=
2p+2.
\]

Hence,
\[
\#\Xi_e^{\mathrm{loc}}=2p+2.
\]
\end{proof}

\begin{remark}[Quadratic example]
\label{rem:quadratic_active_sections}
Consider the quadratic knot vector
\[
\Xi=[0,0,0,0.5,1,1,1],
\qquad p=2.
\]

The two nonzero knot spans are
\[
I_1=[0,0.5],
\qquad
I_2=[0.5,1].
\]

The first Active Section is characterized by
\[
\Xi_1^{\mathrm{loc}}
=
[0,0,0,0.5,1,1],
\]
whereas the second Active Section is characterized by
\[
\Xi_2^{\mathrm{loc}}
=
[0,0,0.5,1,1,1].
\]

Both local knot vectors contain
\[
2p+2=6
\]
entries. The common interface is located at
\[
\xi=0.5.
\]
\end{remark}

\begin{remark}[Distinction from the interface coupling space]
\label{rem:2p2_distinction}
The relation
\[
\#\Xi_e^{\mathrm{loc}}=2p+2
\]
describes the local knot vector of one Active Section.

Later, when two neighboring Active Sections are considered before
coupling, their disconnected approximation spaces contain
\(p+1\) basis functions each and therefore have combined dimension
\[
2(p+1)=2p+2.
\]

The two occurrences of \(2p+2\) are numerically identical but represent
different mathematical statements.
\end{remark}

\subsection{Exact geometry preservation under local unclamped decomposition}
\label{app:exact_local_decomposition}

Let the parent NURBS curve be
\begin{equation}
\label{eq:parent_nurbs}
\bm{x}(\xi)
=
\frac{
\displaystyle
\sum_{i=1}^{n}
N_{i,p}(\xi)w_i\bm{P}_i
}{
\displaystyle
\sum_{i=1}^{n}
N_{i,p}(\xi)w_i
},
\qquad
\xi\in[\xi_{\min},\xi_{\max}],
\end{equation}
where \(N_{i,p}\) are degree-\(p\) B-spline basis functions,
\(\bm{P}_i\in\mathbb{R}^{d}\), and \(w_i>0\).

Let
\[
I_e=[\widehat{\xi}_e,\widehat{\xi}_{e+1}],
\qquad
e=1,\ldots,m,
\]
denote the nonzero knot spans, and define the active index set
\[
\mathcal{I}_e
=
\left\{
i:
N_{i,p}|_{I_e}\not\equiv0
\right\}.
\]

For every nonzero knot span, construct a local unclamped NURBS
representation by extracting:

\begin{enumerate}[label=(\roman*)]
    \item the \(p+1\) parent basis functions active on the span;
    \item their corresponding homogeneous control points;
    \item the inherited local knot vector of length \(2p+2\);
    \item no additional artificial endpoint repetitions.
\end{enumerate}

\begin{theorem}[Exact preservation under local unclamped decomposition]
\label{thm:exact_local_decomposition}
Let \(\bm{x}^{(e)}\) denote the NURBS mapping of Active Section \(e\).
If the Active Section is constructed from the restrictions of the parent
basis functions, their corresponding homogeneous control points, and the
inherited local knot subsequence, then
\[
\boxed{
\bm{x}^{(e)}(\xi)=\bm{x}(\xi)
\qquad
\forall \xi\in I_e.
}
\]

Consequently, the complete piecewise decomposed mapping satisfies
\[
\boxed{
\bm{x}_{\mathrm{dec}}(\xi)=\bm{x}(\xi)
\qquad
\forall
\xi\in[\xi_{\min},\xi_{\max}].
}
\]
\end{theorem}

\begin{proof}
Introduce the homogeneous control points
\[
\bm{Q}_i
=
\begin{bmatrix}
w_i\bm{P}_i\\
w_i
\end{bmatrix}
\in\mathbb{R}^{d+1},
\]
and define the homogeneous parent mapping
\[
\widetilde{\bm{x}}(\xi)
=
\sum_{i=1}^{n}
N_{i,p}(\xi)\bm{Q}_i.
\]

For \(\xi\in I_e\), all basis functions whose indices do not belong to
\(\mathcal{I}_e\) vanish. Hence,
\[
\widetilde{\bm{x}}(\xi)
=
\sum_{i\in\mathcal{I}_e}
N_{i,p}(\xi)\bm{Q}_i.
\]

Let \(N^{(e)}_{a,p}\) and \(\bm{Q}^{(e)}_a\) denote the extracted local
basis functions and homogeneous control points. By construction, there
is a one-to-one correspondence
\[
a\longleftrightarrow i_a\in\mathcal{I}_e
\]
such that
\[
N^{(e)}_{a,p}(\xi)
=
N_{i_a,p}(\xi),
\qquad
\bm{Q}^{(e)}_a=\bm{Q}_{i_a},
\qquad
\xi\in I_e.
\]

Therefore,
\[
\begin{aligned}
\widetilde{\bm{x}}^{(e)}(\xi)
&=
\sum_a
N^{(e)}_{a,p}(\xi)\bm{Q}^{(e)}_a
\\
&=
\sum_{i\in\mathcal{I}_e}
N_{i,p}(\xi)\bm{Q}_i
\\
&=
\widetilde{\bm{x}}(\xi).
\end{aligned}
\]

Projective division by the last homogeneous coordinate gives
\[
\bm{x}^{(e)}(\xi)=\bm{x}(\xi)
\qquad
\forall\xi\in I_e.
\]

Since the nonzero knot spans cover the complete parametric domain, the
piecewise decomposed representation and the parent geometry coincide
everywhere.
\end{proof}

\begin{corollary}[Patchwise degree elevation before extraction]
\label{cor:degree_elevation_before_extraction}
Suppose that, for Active Section \(e\), the parent representation is
first degree-elevated exactly and the local Active Section is subsequently
extracted from the elevated parent representation. Then
\[
\bm{x}^{(e)}(\xi)=\bm{x}(\xi),
\qquad
\xi\in I_e.
\]
\end{corollary}

\begin{proof}
Exact NURBS degree elevation changes the spline representation but leaves
the geometric mapping invariant. Theorem~\ref{thm:exact_local_decomposition}
can therefore be applied to the degree-elevated parent representation.
\end{proof}

\begin{remark}
The extracted local basis functions need not form a globally clamped
basis and need not coincide with the complete parent basis over the
entire parametric domain. It is sufficient that, on the associated
nonzero knot span, they coincide with the restrictions of the
corresponding parent basis functions.
\end{remark}

\subsection{Pairwise interface coupling space}
\label{app:pairwise_interface_space}

Consider two neighboring Active Sections associated with the nonzero
knot spans
\[
I_L=[\xi_{e-1},\xi_e],
\qquad
I_R=[\xi_e,\xi_{e+1}],
\]
which share the interface \(\xi=\xi_e\).

Let \(V_L\) and \(V_R\) denote their disconnected local approximation
spaces before continuity constraints are imposed.

\begin{definition}[Disconnected interface coupling space]
\label{def:interface_coupling_space}
The disconnected interface coupling space is defined as
\[
\mathcal{A}_{LR}
=
V_L\oplus V_R.
\]
It contains the complete local approximation spaces of the two
neighboring Active Sections before coupling.
\end{definition}

\begin{proposition}[Dimension of the disconnected interface space]
\label{prop:interface_space_dimension}
If both neighboring Active Sections have degree \(p\) and their local
coefficient sets are independent before coupling, then
\[
\boxed{
\dim(\mathcal{A}_{LR})=2p+2.
}
\]
\end{proposition}

\begin{proof}
Each degree-\(p\) Active Section possesses \(p+1\) locally active basis
functions. Therefore,
\[
\dim(V_L)=p+1,
\qquad
\dim(V_R)=p+1.
\]

Before coupling, the local coefficient blocks are independent, so
\[
V_L\cap V_R=\{\bm{0}\}.
\]

Consequently,
\[
\begin{aligned}
\dim(\mathcal{A}_{LR})
&=
\dim(V_L\oplus V_R)
\\
&=
\dim(V_L)+\dim(V_R)
\\
&=
(p+1)+(p+1)
\\
&=
2p+2.
\end{aligned}
\]
\end{proof}

\begin{remark}
Proposition~\ref{prop:interface_space_dimension} does not define an
Active Section. It characterizes the disconnected pairwise space used
during the coupling of two already defined Active Sections.
\end{remark}

\subsection{Null-space criterion for exact coupling}
\label{app:nullspace_exact_coupling}

Collect all disconnected local rational basis functions in
\[
\bm{R}(\xi)
=
\begin{bmatrix}
\bm{R}^{(1)}(\xi)
&
\cdots
&
\bm{R}^{(m)}(\xi)
\end{bmatrix},
\]
and define the stacked local control-point matrix
\[
\bm{P}_{\mathrm{loc}}
=
\begin{bmatrix}
\bm{P}^{(1)}
\\
\vdots
\\
\bm{P}^{(m)}
\end{bmatrix}.
\]

The decomposed geometry is written as
\[
\bm{x}_{\mathrm{dec}}(\xi)
=
\bm{R}(\xi)\bm{P}_{\mathrm{loc}}.
\]

Let \(\bm{C}\) denote the global interface constraint matrix, and let the
full-column-rank matrix \(\bm{T}\) satisfy
\[
\operatorname{range}(\bm{T})
=
\ker(\bm{C}).
\]

Define the reconstructed hybrid basis by
\[
\bm{H}(\xi)
=
\bm{R}(\xi)\bm{T}.
\]

\begin{proposition}[Exact-coupling criterion]
\label{prop:exact_coupling}
Assume that
\[
\bm{C}\bm{P}_{\mathrm{loc}}=\bm{0}
\]
and
\[
\operatorname{range}(\bm{T})
=
\ker(\bm{C}).
\]

Then there exists a hybrid control-point matrix
\(\bm{P}_h\) such that
\[
\bm{P}_{\mathrm{loc}}
=
\bm{T}\bm{P}_h.
\]

Consequently,
\[
\boxed{
\bm{H}(\xi)\bm{P}_h
=
\bm{R}(\xi)\bm{P}_{\mathrm{loc}}
=
\bm{x}_{\mathrm{dec}}(\xi).
}
\]
\end{proposition}

\begin{proof}
Write the local control-point matrix columnwise as
\[
\bm{P}_{\mathrm{loc}}
=
\begin{bmatrix}
\bm{p}_1
&
\cdots
&
\bm{p}_d
\end{bmatrix}.
\]

The relation
\[
\bm{C}\bm{P}_{\mathrm{loc}}=\bm{0}
\]
implies that
\[
\bm{C}\bm{p}_{\alpha}=\bm{0},
\qquad
\alpha=1,\ldots,d.
\]

Therefore,
\[
\bm{p}_{\alpha}\in\ker(\bm{C})
=
\operatorname{range}(\bm{T}).
\]

For each coordinate direction, there exists a vector
\(\bm{q}_{\alpha}\) such that
\[
\bm{p}_{\alpha}
=
\bm{T}\bm{q}_{\alpha}.
\]

Collecting these vectors as columns of
\[
\bm{P}_h
=
\begin{bmatrix}
\bm{q}_1
&
\cdots
&
\bm{q}_d
\end{bmatrix}
\]
gives
\[
\bm{P}_{\mathrm{loc}}
=
\bm{T}\bm{P}_h.
\]

It follows that
\[
\begin{aligned}
\bm{H}(\xi)\bm{P}_h
&=
\bm{R}(\xi)\bm{T}\bm{P}_h
\\
&=
\bm{R}(\xi)\bm{P}_{\mathrm{loc}}
\\
&=
\bm{x}_{\mathrm{dec}}(\xi).
\end{aligned}
\]
\end{proof}

\subsection{Recursive locality-preserving pairwise coupling}
\label{app:recursive_pairwise}

Let \(\bm{T}_k\) denote the transformation obtained after coupling Active
Sections \(1,\ldots,k\).

For the first interface, let \(\bm{C}_{12}\) denote the corresponding
constraint matrix, and choose \(\bm{T}_2\) such that
\[
\operatorname{range}(\bm{T}_2)
=
\ker(\bm{C}_{12}).
\]

For \(k\ge3\), define the augmented transformation
\[
\bm{T}_{\mathrm{aug},k}
=
\operatorname{blkdiag}
\left(
\bm{T}_{k-1},
\bm{I}_{n_k}
\right),
\]
where \(n_k\) is the number of local basis functions of Active Section
\(k\).

Let \(\bm{C}_k\) represent the new interface constraints between the
already reconstructed block \(1{:}k-1\) and Active Section \(k\).
Define the reduced constraint matrix
\[
\widehat{\bm{C}}_k
=
\bm{C}_k\bm{T}_{\mathrm{aug},k}.
\]

Let \(\bm{T}_{\mathrm{red},k}\) satisfy
\[
\operatorname{range}(\bm{T}_{\mathrm{red},k})
=
\ker(\widehat{\bm{C}}_k),
\]
and update
\[
\bm{T}_k
=
\bm{T}_{\mathrm{aug},k}
\bm{T}_{\mathrm{red},k}.
\]

\begin{theorem}[Exact geometry preservation under recursive pairwise coupling]
\label{thm:recursive_exact_geometry}
Assume that:

\begin{enumerate}[label=(\roman*)]
    \item
    \[
    \operatorname{range}(\bm{T}_2)
    =
    \ker(\bm{C}_{12});
    \]

    \item for every \(k\ge3\),
    \[
    \operatorname{range}(\bm{T}_{\mathrm{red},k})
    =
    \ker(\widehat{\bm{C}}_k);
    \]

    \item the decomposed geometry satisfies every imposed interface
    condition;

    \item no direction required to span any relevant null space is
    removed during the reconstruction.
\end{enumerate}

Then, for every \(k=2,\ldots,m\), there exists a hybrid control-point
matrix \(\bm{P}_{h,k}\) such that
\[
\bm{P}_{1:k}
=
\bm{T}_k\bm{P}_{h,k},
\]
where
\[
\bm{P}_{1:k}
=
\begin{bmatrix}
\bm{P}^{(1)}
\\
\vdots
\\
\bm{P}^{(k)}
\end{bmatrix}.
\]

In particular,
\[
\bm{P}_{\mathrm{loc}}
=
\bm{T}_m\bm{P}_{h,m},
\]
and therefore
\[
\boxed{
\bm{x}_h(\xi)
=
\bm{x}_{\mathrm{dec}}(\xi)
=
\bm{x}_{\mathrm{original}}(\xi).
}
\]
\end{theorem}

\begin{proof}
The proof proceeds by induction.

For the first pair of Active Sections, geometric consistency gives
\[
\bm{C}_{12}\bm{P}_{1:2}=\bm{0}.
\]

Thus,
\[
\bm{P}_{1:2}
\in
\ker(\bm{C}_{12})
=
\operatorname{range}(\bm{T}_2).
\]

Therefore, there exists a matrix \(\bm{P}_{h,2}\) such that
\[
\bm{P}_{1:2}
=
\bm{T}_2\bm{P}_{h,2}.
\]

Assume now that the statement holds after coupling the first \(k-1\)
Active Sections:
\[
\bm{P}_{1:k-1}
=
\bm{T}_{k-1}\bm{P}_{h,k-1}.
\]

Introduce the augmented coefficient matrix
\[
\bm{P}_{\mathrm{aug},k}
=
\begin{bmatrix}
\bm{P}_{h,k-1}
\\
\bm{P}^{(k)}
\end{bmatrix}.
\]

By the block-diagonal form of
\(\bm{T}_{\mathrm{aug},k}\),
\[
\bm{T}_{\mathrm{aug},k}
\bm{P}_{\mathrm{aug},k}
=
\begin{bmatrix}
\bm{T}_{k-1}\bm{P}_{h,k-1}
\\
\bm{P}^{(k)}
\end{bmatrix}
=
\bm{P}_{1:k}.
\]

Thus, augmentation preserves exact representability.

Since the decomposed geometry satisfies the new interface condition,
\[
\bm{C}_k\bm{P}_{1:k}=\bm{0}.
\]

Substituting the augmented representation yields
\[
\begin{aligned}
\widehat{\bm{C}}_k
\bm{P}_{\mathrm{aug},k}
&=
\bm{C}_k
\bm{T}_{\mathrm{aug},k}
\bm{P}_{\mathrm{aug},k}
\\
&=
\bm{C}_k\bm{P}_{1:k}
\\
&=
\bm{0}.
\end{aligned}
\]

Hence,
\[
\bm{P}_{\mathrm{aug},k}
\in
\ker(\widehat{\bm{C}}_k)
=
\operatorname{range}(\bm{T}_{\mathrm{red},k}).
\]

Therefore, there exists a matrix \(\bm{P}_{h,k}\) such that
\[
\bm{P}_{\mathrm{aug},k}
=
\bm{T}_{\mathrm{red},k}
\bm{P}_{h,k}.
\]

Consequently,
\[
\begin{aligned}
\bm{P}_{1:k}
&=
\bm{T}_{\mathrm{aug},k}
\bm{P}_{\mathrm{aug},k}
\\
&=
\bm{T}_{\mathrm{aug},k}
\bm{T}_{\mathrm{red},k}
\bm{P}_{h,k}
\\
&=
\bm{T}_k\bm{P}_{h,k}.
\end{aligned}
\]

The induction is complete. Combining this result with
Theorem~\ref{thm:exact_local_decomposition} gives
\[
\bm{x}_h
=
\bm{x}_{\mathrm{dec}}
=
\bm{x}_{\mathrm{original}}.
\]
\end{proof}

\subsection{Active and frozen coordinates}
\label{app:active_frozen}

During the coupling of the reconstructed block with a new Active
Section, only the coordinates participating in the newly introduced
interface constraints must be recombined.

After a suitable coordinate permutation, write
\[
\widehat{\bm{C}}_k
=
\begin{bmatrix}
\bm{0}
&
\bm{C}_{\mathrm{act}}
\end{bmatrix},
\]
where the zero block corresponds to frozen coordinates and
\(\bm{C}_{\mathrm{act}}\) acts only on active interface coordinates.

\begin{proposition}[Active--frozen null-space decomposition]
\label{prop:active_frozen}
If
\[
\operatorname{range}(\bm{T}_{\mathrm{act}})
=
\ker(\bm{C}_{\mathrm{act}}),
\]
then
\[
\boxed{
\ker(\widehat{\bm{C}}_k)
=
\mathbb{R}^{n_f}
\oplus
\ker(\bm{C}_{\mathrm{act}}),
}
\]
where \(n_f\) is the number of frozen coordinates.

Up to the coordinate permutation, a corresponding reconstruction matrix
is
\[
\bm{T}_{\mathrm{red},k}
=
\begin{bmatrix}
\bm{I}_{n_f} & \bm{0}
\\
\bm{0} & \bm{T}_{\mathrm{act}}
\end{bmatrix}.
\]

Thus, frozen coordinates pass unchanged, while only the coordinates
affected by the new interface constraints are reconstructed.
\end{proposition}

\begin{proof}
Let
\[
\bm{z}
=
\begin{bmatrix}
\bm{z}_f
\\
\bm{z}_a
\end{bmatrix}
\]
be partitioned into frozen and active components.

Then
\[
\widehat{\bm{C}}_k\bm{z}
=
\begin{bmatrix}
\bm{0}
&
\bm{C}_{\mathrm{act}}
\end{bmatrix}
\begin{bmatrix}
\bm{z}_f
\\
\bm{z}_a
\end{bmatrix}
=
\bm{C}_{\mathrm{act}}\bm{z}_a.
\]

Therefore,
\[
\widehat{\bm{C}}_k\bm{z}=\bm{0}
\]
if and only if
\[
\bm{C}_{\mathrm{act}}\bm{z}_a=\bm{0}.
\]

The frozen component \(\bm{z}_f\) is arbitrary, whereas
\[
\bm{z}_a\in\ker(\bm{C}_{\mathrm{act}}).
\]

Hence,
\[
\ker(\widehat{\bm{C}}_k)
=
\mathbb{R}^{n_f}
\oplus
\ker(\bm{C}_{\mathrm{act}}).
\]

Since
\[
\operatorname{range}(\bm{T}_{\mathrm{act}})
=
\ker(\bm{C}_{\mathrm{act}}),
\]
the stated block-diagonal transformation spans the complete reduced null
space.
\end{proof}

\subsection{Recovery of hybrid control variables}
\label{app:pseudoinverse_recovery}

\begin{corollary}[Recovery by the Moore--Penrose pseudoinverse]
\label{cor:pseudoinverse}
Once
\[
\bm{P}_{\mathrm{loc}}
\in
\operatorname{range}(\bm{T}_m)
\]
has been established, the hybrid control-point matrix may be computed as
\[
\bm{P}_h
=
\bm{T}_m^{+}\bm{P}_{\mathrm{loc}},
\]
where \(\bm{T}_m^{+}\) denotes the Moore--Penrose pseudoinverse.

Then
\[
\bm{T}_m\bm{T}_m^{+}\bm{P}_{\mathrm{loc}}
=
\bm{P}_{\mathrm{loc}},
\]
and the reconstructed geometry is exact.
\end{corollary}

\begin{proof}
The matrix
\[
\bm{T}_m\bm{T}_m^{+}
\]
is the orthogonal projector onto
\(\operatorname{range}(\bm{T}_m)\).

Since
\[
\bm{P}_{\mathrm{loc}}
\in
\operatorname{range}(\bm{T}_m),
\]
projection leaves every column of
\(\bm{P}_{\mathrm{loc}}\) unchanged. Therefore,
\[
\bm{T}_m\bm{T}_m^{+}\bm{P}_{\mathrm{loc}}
=
\bm{P}_{\mathrm{loc}}.
\]
\end{proof}

\begin{remark}
The pseudoinverse does not establish exact geometry preservation by
itself. It recovers the exact hybrid coefficients only after the range
inclusion
\[
\bm{P}_{\mathrm{loc}}
\in
\operatorname{range}(\bm{T}_m)
\]
has been proved.
\end{remark}

\begin{remark}
Positivity, partition of unity, locality, linear independence, and
conditioning are desirable properties of the reconstructed basis.
Exact geometry preservation, however, follows specifically from the
null-space and range relations established above.
\end{remark}

\qquad


\section{  Proof of the Minimal Local Knot Vector Property    }
\label{app:B}


This appendix establishes the mathematical foundation of the local
\(2p+2\) knot vector associated with every Active Section.
The result is independent of the reconstruction procedure and follows
directly from the local support properties of B-spline basis functions.

\setcounter{subsection}{0}
\renewcommand{\thesubsection}{B.\arabic{subsection}}


\subsection{Active Basis over a Single Knot Span}
\label{app:B_active_basis}

Let
\[
\Xi=
\{\xi_0,\xi_1,\ldots,\xi_{n+p+1}\}
\]
be the knot vector of a univariate B-spline or NURBS representation of
degree \(p\).

Consider a nonzero knot span

\[
[\xi_k,\xi_{k+1}].
\]

It is well known that exactly \(p+1\) basis functions are nonzero over
this interval,

\[
N_{k-p,p},
N_{k-p+1,p},
\ldots,
N_{k,p}.
\]

Consequently, every approximation over the selected knot span depends
exclusively on these \(p+1\) basis functions.

\subsection{Local knot vector of an individual basis function}

Each degree-\(p\) B-spline basis function

\[
N_{i,p}
\]

is completely determined by its local knot vector

\[
[\xi_i,\xi_{i+1},\ldots,\xi_{i+p+1}],
\]

which contains exactly

\[
p+2
\]

knot entries.

\subsection{Minimal knot vector of a knot span}

To evaluate all basis functions active on the knot span
\([\xi_k,\xi_{k+1}]\),
one must retain the union of the local knot vectors associated with

\[
N_{k-p,p},
\ldots,
N_{k,p}.
\]

Therefore,

\[
\bigcup_{i=k-p}^{k}
[\xi_i,\ldots,\xi_{i+p+1}]
=
[\xi_{k-p},\ldots,\xi_{k+p+1}].
\]

The number of knot entries is

\[
(k+p+1)-(k-p)+1
=
2p+2.
\]

Hence the local spline description associated with a single knot span
requires exactly

\[
\boxed{2p+2}
\]

knot entries.

\subsection{Minimality}

The obtained knot vector is minimal.

Indeed, removing any knot from

\[
[\xi_{k-p},\ldots,\xi_{k+p+1}]
\]

eliminates one endpoint of the local knot vector of at least one active
basis function. Consequently, at least one of the \(p+1\) basis
functions can no longer be evaluated correctly over the selected knot
span.

Therefore no shorter knot vector contains all information required for
the exact local representation.

\subsection{Relation to the Active Section}

The proposed decomposition associates one Active Section with every
nonzero knot span.

Since each Active Section contains exactly the \(p+1\) basis functions
active on its knot span, its local spline description is completely
determined by the minimal local knot vector

\[
[\xi_{k-p},\ldots,\xi_{k+p+1}],
\]

which contains precisely

\[
2p+2
\]

knot entries.

This property is independent of the reconstruction algorithm and depends
only on the local support of B-spline basis functions.

\begin{remark}
The quantity \(2p+2\) characterizes the minimal local knot vector of a
single Active Section.
This should not be confused with the larger knot neighborhood obtained
when two adjacent Active Sections are considered simultaneously during a
pairwise reconstruction step. The latter is a different mathematical
object and is not used in the definition of an Active Section.
\end{remark}

\end{document}